\documentclass[12pt,letter]{article}
\usepackage{amsmath}
\usepackage{graphicx}
\usepackage{enumerate}
\usepackage{natbib}
\usepackage{url} 

\usepackage[utf8]{inputenc}
\usepackage[english]{babel}
\usepackage{bm}
\usepackage{amsfonts}
\usepackage{multirow}
\usepackage{linguex}
\usepackage{lscape}
\usepackage{booktabs}
\usepackage{amssymb,amsthm}
\usepackage{array}
\usepackage{paralist}
\usepackage{verbatim}
\usepackage{subfig}
\usepackage[usenames, dvipsnames]{color}
\usepackage{authblk}
\usepackage{thmtools}
\usepackage{chngcntr}
\usepackage{setspace}

\usepackage{ulem}

\newcommand{\blind}{1} 

\newcommand{\mR}{\mathbb{R}}
\newcommand{\Pn}{\mathbb{P}_n}
\newcommand{\bX}{\mathbf{X}}
\newcommand{\bI}{\mathbf{I}}
\newcommand{\bA}{\mathbf{A}}
\newcommand{\bxi}{\bm{\xi}}
\newcommand{\bbeta}{\bm{\beta}}
\newcommand{\bx}{\bm{x}}
\newcommand{\bW}{\bm{W}}
\newcommand{\bTheta}{\bm{\Theta}}
\newcommand{\bSigma}{\bm{\Sigma}}

\newcommand{\E}{\mathbb{E}}
\newcommand{\lambdamax}{\lambda_{\mathrm{max}}}
\newcommand{\lambdamin}{\lambda_{\mathrm{min}}}
\newcommand{\ba}{\bm{a}}
\newcommand{\bb}{\bm{b}}
\newcommand{\bz}{\bm{z}}

\newcommand{\OP}{\mathcal{O}_{\mathbb{P}}}
\newcommand{\oP}{o_{\mathbb{P}}}
\newcommand{\balpha}{\bm{\alpha}}
\newcommand{\bDelta}{\bm{\Delta}}

\newtheorem{theorem}{Theorem}

\newtheorem{lemma}{Lemma}
\newtheorem{remark}{Remark} 

\begin{document}

	\def\spacingset#1{\renewcommand{\baselinestretch}%
		{#1}\small\normalsize} \spacingset{1}


	
	\if1\blind
	{
		\title{\bf A Revisit to De-biased Lasso for Generalized Linear Models}
		\author{Lu Xia$^1$
			\thanks{Supported in part by NIH R01 AG056764},
			Bin Nan$^2$
			\thanks{Supported in part by NIH R01 AG056764 and NSF DMS-1915711}
			~ and Yi Li$^1$
			\thanks{Supported in part by NIH R01 AG056764 and U01CA209414} 
			\\
			$^1$ Department of Biostatistics, University of Michigan, Ann Arbor \\
			$^2$ Department of Statistics, University of California, Irvine 
		}
		\date{}
		\maketitle
	} \fi
	
	\if0\blind
	{
		\bigskip
		\bigskip
		\bigskip
		\title{\bf A Revisit to De-biased Lasso for Generalized Linear Models}
		\author{}
		\date{}
		\maketitle
	} \fi
	
	\bigskip
	\begin{abstract}
		De-biased lasso  has emerged as a popular tool to draw statistical inference for high-dimensional regression models. However, simulations indicate that for generalized linear models (GLMs), de-biased lasso  inadequately removes biases and yields unreliable confidence intervals. This motivates us to scrutinize the application of de-biased lasso in high-dimensional GLMs. When $p >n$, we detect that a key sparsity condition on the inverse information matrix generally does not hold in a GLM setting, which likely explains the subpar performance of de-biased lasso. Even in a less challenging ``large $n$, diverging $p$" scenario, we find that de-biased lasso and the maximum likelihood method often yield confidence intervals with unsatisfactory coverage probabilities. In this scenario, we examine an alternative approach for further bias correction by directly inverting the Hessian matrix without imposing the matrix sparsity assumption. We
		establish the asymptotic distributions of any linear combinations of the resulting estimates, which lay the theoretical groundwork for drawing inference. Simulations show that this refined de-biased estimator performs well in removing biases and yields an honest confidence interval coverage. We illustrate the method by analyzing a prospective hospital-based Boston Lung Cancer Study, a large scale epidemiology cohort investigating the joint effects of genetic variants on lung cancer risk. \\
		
		\noindent%
	\end{abstract}
	
	\newpage

	\spacingset{1.45} 
	\section{Introduction}
	\label{sec:intro}
	
	
	Traditional  genome-wide association studies typically screen marginal associations between single nucleotide polymorphisms (SNPs) and complex traits. However, the marginal approach does not take into account the complicated structural relationships among SNPs. Jointly modeling the effects of SNPs within target genes  can pinpoint  functionally impactful loci in the coding regions \citep{taylor2001using, repapi2010genome}, better understand the molecular mechanisms underlying complex diseases \citep{guan2011bayesian}, reduce false positives around true causal SNPs and improve prediction accuracy \citep{he2010variable}. For example, in the Boston Lung Cancer Study (BLCS), which investigates molecular mechanisms underlying the lung cancer, an analytical goal is to study  the joint effects of genetic variants residing in multiple disease related pathway genes on lung cancer risk. The results can potentially aid personalized medicine as individualized therapeutic interventions are only possible with proper characterization of relevant SNPs in pharmacogenomics \citep{evans2004moving}. Statistically, this requires reliable inference on  high-dimensional regression models.

	It is of great interest, though with enormous challenges, to draw  inference when the number of covariates grows with the sample size. When the number of covariates  exceeds the sample size, the well known ``large $p$, small $n$" scenario,  maximum likelihood estimation (MLE) is no longer feasible and regularized variable selection methods have been developed over the decades. These include the lasso method  \citep{tibshirani1996regression}, the elastic net method  \citep{zou2005regularization}, and the Dantzig selector \citep{candes2007dantzig}, among many others. However,  these regularized methods yield biased estimates, and thus cannot be directly used for drawing statistical inference, in particular,  constructing confidence intervals with a nominal coverage. Even when the number of covariates is smaller than the sample size but can increase with $n$, conventional methods may still not be trustworthy.  \cite{sur2019modern} showed that  MLE for high-dimensional logistic regression models can overestimate the magnitudes of non-zero effects while underestimating the variances of the estimates when the number of covariates is smaller than, but of the same order as, the sample size. We encountered the same difficulty when applying  MLE to the analysis of BLCS data.

	Advances to address these challenges have been made recently. One stream of methods is post-selection inference conditional on selected models \citep{lee2016exact}, which ignores the uncertainty associated with model selection. Other super-efficient procedures, such as SCAD \citep{fan2001variable} and adaptive lasso \citep{zou2006adaptive}, share the flavor of post-selection inference.  Another school of methods is to draw inference by de-biasing the lasso estimator, termed de-biased lasso or de-sparsified lasso, which  relieves the restrictions of post-selection inference and has been shown to possess nice theoretical and numerical properties in linear regression models (\citealt{van2014asymptotically, zhang2014confidence, javanmard2014confidence}). When coefficients have group structures, various extensions of de-biased lasso have been proposed \citep{zhang2017simultaneous, dezeure2017high, mitra2016benefit, cai2019sparse}.
	
	De-biased lasso has seen applications beyond linear models.  For example, \citet{van2014asymptotically} considered the de-biased lasso approach in generalized linear models (GLMs) and developed the asymptotic normality theory for each component of the coefficient estimates; \citet{zhang2017simultaneous} proposed a multiplier bootstrap procedure to draw inference on a group of coefficients in GLMs, yet without sufficient numerical evidence for the performance; \citet{eftekhari2019inference} considered a de-biased lasso estimator for a low-dimensional component in a generalized single-index model with an unknown link function and  restricted to an elliptically symmetric design.

	However, in the GLM setting, our extensive simulations reveal that biases cannot be adequately removed by the existing de-biased lasso methods. Even after de-biasing, the biases are still too large relative to the model based standard errors, and the resulting confidence intervals have much lower coverage probabilities than the nominal level. Scrutiny of the existing theories points to a key assumption: the inverse of the Fisher information matrix is sparse (see \citealt{van2014asymptotically}). For linear regression, this assumption amounts to that the precision matrix for the covariates is sparse.
	It, however, is unlikely to hold in GLM settings, even when the precision matrix for the covariates is indeed sparse.
	
	This begs a critical question: when can we obtain reliable inference results using de-biased lasso? Deviated from the aforementioned works which mainly focused on hypothesis testing, we are concerned with making reliable inference, such as eliminating estimation bias and obtaining good confidence interval coverage. 
	We consider two scenarios: the ``large $p$, small $n$" case where $p > n$, and the ``large $n$, diverging $p$" case where $p$ increases to infinity with $n$ but $p/n \rightarrow 0$. In the first scenario, we discuss a key sparsity assumption in GLMs, which is  likely to fail and  compromise the validity of de-biased lasso. In the second scenario, we consider a natural alternative for further bias correction, by directly inverting the Hessian matrix. We study its theoretical properties and use simulations to demonstrate its advantageous performance to the competitors. 
	
	The remainder of the paper is organized as follows. Section 2 briefly reviews de-biased lasso in GLMs. In Section 3, we exemplify the performance of the original de-biased lasso estimator using simulated examples and elaborate on the theoretical limitations. In Section 4, under the ``large $n$, diverging $p$" regime, we introduce a refined de-biased approach as an alternative to the node-wise lasso estimator for the inverse of the information matrix \citep{van2014asymptotically}, and establish asymptotic distributions for any linear combinations of the refined de-biased estimates. We  provide simulation results and analyze the Boston Lung Cancer Study that investigates the joint associations of SNPs in nine candidate genes with lung cancer. We conclude with the summarized findings in Section 5. Additional numerical results are provided in the online supplementary material.

	\section{Background}

	\subsection{Notation}
	
	We define  commonly used notation. Denote by $\lambdamax$ and $\lambdamin$  the largest and the smallest eigenvalue of a symmetric matrix.
	For a real matrix $\bA = (A_{ij})$, let $\| \bA \| = [\lambdamax(\bA^T \bA)]^{1/2}$ be the spectral norm. The induced matrix $\ell_1$ norm is $\|\bA\|_1 = \max_j \sum_i  |A_{ij}|$, and  when $\bA$ is symmetric,  $\|\bA\|_1 = \max_i \sum_j  |A_{ij}|$. The entrywise $\ell_{\infty}$ norm is $\|\bA\|_{\infty} = \max_{i,j} |A_{ij}|$. For a vector $\ba$, $\|\ba\|_q$ denotes the $\ell_q$ norm, $q \ge 1$. We write $x_n \asymp y_n$ if $x_n = \mathcal{O}(y_n)$ and $y_n = \mathcal{O}(x_n)$. 
	
	\subsection{Generalized linear models}
	
	Denote by $y_i$ the response variable and $\bx_i = (1, \widetilde{\bx}_i^T)^T \in \mathbb{R}^{p+1}$ for $i = 1,\cdots, n$, where the first element in $\bx_i$ corresponds to the intercept, and the rest elements $\widetilde{\bx}_i$ represent $p$ covariates. Let $\bX$ be an $n \times (p+1)$ covariate matrix with $\bx_i^T$ being the $i$th row.
	We assume  that $\left\{(y_i, \bx_i)\right\}_{i=1}^n$ are independently and identically distributed (i.i.d.) copies of $(y,\bx)$. 
	Define  the negative log-likelihood function (up to a constant irrelevant to the unknown parameters) when the conditional density of $y$ given $\bx$ belongs to the linear exponential family:
	\begin{equation} \label{eq:loss_func}
	\rho_{\bxi}(y, \bx) \equiv \rho(y, \bx^T \bxi) = - y \bx^T \bxi + b(\bx^T \bxi)
	\end{equation}
	where $b(\cdot)$ is a known twice continuously differentiable function, $\bxi = (\beta_0, \bbeta^T)^T \in \mR^{p+1}$ denotes the vector of regression coefficients and  $\beta_0 \in \mR$ is the intercept parameter. The unknown true coefficient vector is  $\bxi^0 = (\beta_0^0, {\bbeta^{0}}^T)^T$.

	\subsection{De-biased lasso}
	
	Consider the loss function $\rho_{\bxi}(y, \bx) \equiv \rho(y, \bx^T \bxi)$ given in (\ref{eq:loss_func}). Denote its first and second order derivatives with respect to $\bxi$ by $\dot{\bm{\rho}}_{\bxi}$ and $\ddot{\bm{\rho}}_{\bxi}$, respectively. For any function $g(y, \bx)$, let $\Pn g = \displaystyle \frac{1}{n} \sum_{i=1}^n g(y_i, \bx_i)$. Then for any $\bxi \in \mathbb{R}^{p+1}$, we denote the empirical loss function based on the random sample $\{ (y_i, \bx_i) \}_{i=1}^n$ by $\Pn\rho_{\bxi} \equiv \displaystyle \frac{1}{n} \sum_{i=1}^n \rho_{\bxi}(y_i, \bx_i)$, and its first and second order derivatives with respect to $\bxi$ by $\Pn \dot{\bm{\rho}}_{\bxi} = \displaystyle \frac{1}{n} \sum_{i=1}^n \frac{\partial \rho_{\bxi}(y_i, \bx_i)}{\partial \bxi}$ and $\widehat{\bSigma}_{\bxi} \equiv \Pn \ddot{\bm{\rho}}_{\bxi} = \displaystyle \frac{1}{n} \sum_{i=1}^n \frac{\partial^2 \rho_{\bxi}(y_i, \bx_i)}{\partial \bxi \partial \bxi^T}$. Two important population-level matrices are the expectation of the Hessian matrix, $\bSigma_{\bxi} \equiv \E\widehat{\bSigma}_{\bxi} = \E (\Pn \ddot{ \bm{\rho}}_{\bxi})$, and its inverse $\bTheta_{{\bxi}} \equiv \bSigma_{\bxi}^{-1}$. With $\lambda > 0$, the lasso estimator for $\bxi^0$ is defined as
	\begin{equation} \label{eq:lasso}
	\widehat{\bxi} = \mathop{\arg\min}_{\bxi = (\beta_0, \, \bbeta^T)^T \in \mR^{p+1}} \left\{ \Pn\rho_{\bxi} + \lambda \| \bbeta \|_1 \right\}.
	\end{equation}
	To avoid ambiguity, we do not penalize the intercept $\beta_0$ in (\ref{eq:lasso}). The theoretical results such as prediction and $\ell_1$ error bounds, however, are the same as those in   \citet{van2008high} and \citet{van2014asymptotically} where all the coefficients are penalized \citep{buhlmann2011statistics}. \cite{van2014asymptotically} applied the node-wise lasso method to obtain an estimator $\widehat{\bTheta}$ for $\bTheta_{\bxi^0}$, and  proposed  a de-biased lasso estimator for $\xi_j^0$ with:
	\[
	\widehat{b}_j \equiv \widehat{\xi}_j - \widehat{\bTheta}_j \Pn \dot{\bm{\rho}}_{\widehat{\bxi}},
	\]
	where  $\hat{\sigma}_j \equiv \sqrt{\widehat{\bTheta}_j \widehat{\bSigma}_{\widehat{\bxi}} \widehat{\bTheta}_j^T / n}$ is the model based standard error for $\widehat{b}_j$. Here, $\widehat{\bTheta}_j$ is the $j$th row of $\widehat{\bTheta}$.

	\section{The ``large $p$, small $n$" scenario}
	\label{sec:largep}
	
	Even though the asymptotic theory has been developed for the ``large $p$, small $n$" scenario \citep{van2014asymptotically}, we examine why  de-biased lasso performs unsatisfactorily in GLMs.

	\subsection{A simulation study}
	
	We present a simulation study that features a logistic regression model with $n=300$ observations and $p=500$ covariates. For simplicity, covariates are simulated from $N_{p}(\mathbf{0}, \bSigma_x)$, where $\bSigma_{x, ij} = 0.7^{|i-j|}$, and truncated at $\pm 6$. In the true coefficient vector $\bbeta^0$, the intercept $\beta^0_0 = 0$ and $\beta^0_1$ varies from 0 to 1.5 with 40 equally spaced increments. To examine the impacts of different true model sizes, we arbitrarily choose 2, 4 or 10 additional coefficients from the rest in $\bbeta^0$, and fix them at 1 throughout the simulation. At each value of $\beta_1^0$, a total of 500 simulated datasets are generated. We focus on the de-biased estimates and inference for $\beta_1^0$. 
	
	Figure \ref{fig:sim_largep}, with the true model size increasing from the top to the bottom, shows that the de-biased lasso estimate for $\beta_1^0$ has a bias almost linearly increasing with the true  size of $\beta_1^0$. This  undermines the credibility of the consequent confidence intervals. Meanwhile, the model-based variance does not approximate the true variance well, overestimating the variance for smaller signals and underestimating for larger ones in the two smaller models, as shown by the top two rows in Figure \ref{fig:sim_largep}. This partially explains the over- and under-coverage for smaller and larger signals, respectively. Due to penalized estimation in $\widehat{\bTheta}$, the  variance of the de-biased lasso estimator is even smaller than the ``Oracle" maximum likelihood estimator obtained as if the true model were known; see the bottom two rows in Figure \ref{fig:sim_largep}. The empirical coverage probability decreases to about 50\% as the signal $\beta_1^0$ goes to 1.5, and when the true model size reaches 5;  see the middle row in Figure \ref{fig:sim_largep}. The bias correction is sensitive to the true model size, which becomes worse for larger true models. We have also conducted simulations by changing the covariance structure of covariates to be  independent or compound symmetry with correlation coefficient 0.7 and variance 1, and  have obtained similar results.

	\subsection{Reflections on the validity of theoretical assumptions}

	\citet{van2014asymptotically} established the asymptotic properties of the de-biased lasso estimator in GLMs under certain regularity conditions (see Section 3 of \citealt{van2014asymptotically}), which are imposed to regularize the behavior of the lasso estimator $\widehat{\bxi}$  and the estimated matrix $\widehat{\bTheta}$. \citet{van2014asymptotically} employed the node-wise lasso estimator for $\bTheta_{{\bxi^0}}$, which was originally proposed by \citet{meinshausen2006high} for covariance selection in high-dimensional graphs.
	
	We now revisit the de-biased lasso estimator and its decomposition. The first order Taylor expansion of $\mathbb{P}_n \dot{{\bm{\rho}}}_{\bxi^0}$ at $\widehat{\bxi}$ gives
	\begin{equation} \label{eq:taylor}
	\mathbb{P}_n \dot{{\bm{\rho}}}_{\bxi^0} = \mathbb{P}_n \dot{\bm{\rho}}_{\widehat{\bxi}} + \mathbb{P}_n \ddot{\bm{\rho}}_{\widehat{\bxi}} (\bxi^0 - \widehat{\bxi}) + \bm{\Delta},
	\end{equation}
	where $\bm{\Delta}$ is a $(p+1)$-dimensional vector of remainder terms with its $j$th element 
	\begin{equation} \label{eq:taylor_remainder}
	\Delta_j  = \displaystyle \frac{1}{n} \sum_{i=1}^n  \left(  \ddot{\rho}(y_i, a_j^*) - \ddot{\rho}(y_i, \bx_i^T \widehat{\bxi}) ~ \right) x_{ij} \bx_i^T (\bxi^0 - \widehat{\bxi}),
	\end{equation}  
	in which  $\ddot{\rho}(y,a) \equiv \displaystyle \frac{\partial^2 \rho(y,a)}{\partial a^2}$, and $a_j^*$ lies between $\bx_i^T\widehat{\bxi}$ and $\bx_i^T\bxi^0$. It follows that $\bm{\Delta} = \bm{0}$ in linear regression models,  but generally non-zero in GLMs. Multiplying both sides of (\ref{eq:taylor}) by $\widehat{\bTheta}_j$ and re-organizing the terms, we obtain the following equality for the $j$th component
	\begin{equation} \label{eq:derive_bhat}
	\left[ ~ \widehat{\xi}_j  + \overbrace{ \left( - \widehat{\bTheta}_j \mathbb{P}_n \dot{\bm{\rho}}_{\widehat{\bxi}}   \right) }^{I_j} 
	+ \overbrace{\left( -\widehat{\bTheta}_j \bm{\Delta} \right) }^{II_j} 
	+ \overbrace{\left( \widehat{\bTheta}_j \mathbb{P}_n \ddot{\bm{\rho}}_{\widehat{\bxi}} - \bm{e}_j^T \right)  \left( \widehat{\bxi} - \bxi^0 \right)}^{III_j}  ~ \right] -  \xi^0_j = - \widehat{\bTheta}_j \mathbb{P}_n \dot{\bm{\rho}}_{{\bxi}^0},
	\end{equation}
	where $\bm{e}_j$ is a $(p+1)$-dimensional vector with the $j$th element being 1 and 0 elsewhere. 
	We define three terms $I_j =  - \widehat{\bTheta}_j \mathbb{P}_n \dot{\bm{\rho}}_{\widehat{\bxi}}$, $II_j = -\widehat{\bTheta}_j \bm{\Delta}$ and $III_j = \left( \widehat{\bTheta}_j \mathbb{P}_n \ddot{\bm{\rho}}_{\widehat{\bxi}} - \bm{e}_j^T \right) \left( \widehat{\bxi} - \bxi^0 \right)$. They are crucial in studying the bias behavior of the de-biased lasso estimator that can be alternatively expressed as $\widehat{b}_j = \widehat{\xi}_j + I_j$. According to  (\ref{eq:derive_bhat}), as long as $\sqrt{n}~ II_j / \hat{\sigma}_j = \oP(1)$, $\sqrt{n}~ III_j /\hat{\sigma}_j  = \oP(1)$,  and  $\sqrt{n}~ \widehat{\bTheta}_j \Pn \dot{\bm{\rho}}_{\bxi^0} / \hat{\sigma}_j $ is asymptotically normal, the asymptotic normality of $\sqrt{n} \left( \widehat{b}_j - \xi_j^0 \right) / \hat{\sigma}_j $ follows directly. 
	
	The de-biased lasso approach requires an appropriate inverse matrix estimator with $\mathcal{O}(p^2)$ unknown parameters. In the ``large $p$, small $n$" scenario, where the number of covariates can be  as large as $o (\exp(n^a))$ for some $a>0$, the $(p+1) \times (p+1)$ inverse information matrix is not estimable without further assumptions on the structure of $\bTheta_{{\bxi^0}}$. This inevitably needs regularization, and $\ell_1$-type regularization is often adopted due to its theoretical readiness. An important assumption on $\bTheta_{{\bxi^0}}$ in \citet{van2014asymptotically} is the $\ell_0$ sparsity, i.e. the number of non-zero elements of each row in $\bTheta_{{\bxi^0}}$ is small. This assumption is vital for the consistency of $\widehat{\bTheta}_j$ to $\bTheta_{{\bxi^0},j}$ and consequently the model-based variance, and impacts the negligibility of  term $III_j$ in (\ref{eq:derive_bhat}). In particular, the third bias term in (\ref{eq:derive_bhat}) $III_j$  is non-negligible if the convergence rate of $\widehat{\bTheta}_j$ to $\bTheta_{\bxi^0,j}$, which depends on the $\ell_0$ sparsity of the row vector $\bTheta_{\bxi^0,j}$ using the node-wise lasso estimation, is not fast enough. 
	
	However,  these  sparsity assumptions have not been  clarified in the existing literature, except for linear regression models. In a linear regression model, $\bTheta_{\bxi^0}$ is the precision matrix for covariates which is free of $\bxi^0$, and for multivariate Gaussian covariates, a zero element of $\bTheta_{\bxi^0}$ implies conditional independence between corresponding covariates. In contrast, the row sparsity assumption on $\bTheta_{\bxi^0}$ does not have a clear interpretation in GLMs, and may not be valid as it depends on the unknown $\bxi^0$. In the information matrix $\bSigma_{\bxi^0}$,  its $(j,k)$-th element is $\E\left[x_{ij} x_{ik} \ddot{\rho}(y_i, \bx_i^T \bxi^0) \right] = \E\left[x_{ij} x_{ik} \ddot{b}(\bx_i^T\bxi^0)\right]$. In the most extreme case where all covariates are independent with mean zero, $\bSigma_{\bxi^0, jk} = 0$ for $j \ne k, j= 2, \cdots, p+1$, $k \in \{ k :  2 \le k \le p+1, \xi^0_k = 0 \}$, and then $\bTheta_{\bxi^0}$ is sparse if the true model $\{ j:  1 \le j \le p, ~ \beta_j^0 \ne 0 \}$ is small.  With covariates generally correlated, it is unconceivable that most off-diagonal elements in $\bTheta_{\bxi^0}$ are zero, because $\ddot{b}(\bx_i^T \bxi^0) = \ddot{b}(\beta^0_0 + \widetilde{\bx}_i^T \bbeta^0)$ also depends on the covariates $\widetilde{\bx}_i$ in a GLM, even when the precision matrix for $\widetilde{\bx}_i$ is sparse \textit{per se}. This makes the sparsity  assumption for $\bTheta_{\bxi^0}$ obscure in GLMs. To see this, consider the Poisson regression, which has a closed-form expression for $\bTheta_{\bxi^0}$. Assume the covariates $\widetilde{\bx}_i \sim N_{p}(\mathbf{0}, \bSigma_x)$ and the mean response conditional on $\widetilde{\bx}_i$ is $\mu_i = \exp\{ \beta_0^0 + \widetilde{\bx}_i^T \bbeta^0 \}$ under the canonical link. Then, we have
	\begin{equation*}
	\bSigma_{\bxi^0} = 
	\exp \left\{ \beta_0^0 + \frac{1}{2}{\bbeta^0}^T \bSigma_x \bbeta^0 \right\} \left(
	\begin{array}{cc}
	1 & {\bbeta^0}^T \bSigma_x \\
	\bSigma_x \bbeta^0 & \bSigma_x + \bSigma_x \bbeta^0 {\bbeta^0}^T \bSigma_x \\
	\end{array}
	\right)
	\end{equation*}
	and 
	\begin{equation*} \label{eq:true_theta}
	\bTheta_{\bxi^0} = \exp \left\{ - \beta_0^0 - \frac{1}{2}{\bbeta^0}^T \bSigma_x \bbeta^0 \right\}
	\left(
	\begin{array}{cc}
	\displaystyle \frac{1}{c} & \displaystyle -\frac{1}{c} \ba^T \bA^{-1}  \\
	\displaystyle -\frac{1}{c}  \bA^{-1} \ba & \bA^{-1} + \displaystyle \frac{1}{c} \bA^{-1} \ba \ba^T \bA^{-1} \\
	\end{array}
	\right),
	\end{equation*}
	where $\bA = \bSigma_x + \bSigma_x \bbeta^0 {\bbeta^0}^T \bSigma_x$, $\ba = \bSigma_x \bbeta^0$ and $c = 1 - {\bbeta^0}^T (\bSigma_x^{-1} + \bbeta^0 {\bbeta^0}^T )^{-1} \bbeta^0$. In an over-simplified case where covariates are independent ($\bSigma_x = \bI_p$) and $\bbeta^0$ is  sparse, $\bA^{-1} + \displaystyle \frac{1}{c} \bA^{-1} \ba \ba^T \bA^{-1}$ can be a sparse matrix. However, with often complicated correlation structures between covariates, signal positions and strengths in $\bbeta^0$, it is difficult to guarantee that $\bA^{-1} + \displaystyle \frac{1}{c} \bA^{-1} \ba \ba^T \bA^{-1}$ is sparse.

	To summarize, we believe that the sparsity assumption imposed on $\bTheta_{\bxi^0}$ plays an extremely important role in obtaining the desirable asymptotic properties and finite sample performance of  de-biased lasso in GLMs. However, such an assumption is hardly justifiable in a GLM setting. As evidenced by our simulations, the gap between theory and practice  likely explains the problematic performance of  de-biased lasso  in the ``large $p$, small $n$" scenario. Also note that both bias terms $II_j$ and $III_j$ are not even computable and cannot be recovered, because they involve the unknown $\bxi^0$. All point to that  de-biased lasso generally does not  work well in GLMs in the ``large $p$, small $n$" scenario. 
	
	\section{The ``large $n$, diverging $p$" scenario}
	
	We next study de-biased lasso in GLMs when $p < n$ but $p$ diverges to infinity with $n$ by eliminating more biases, where, under certain conditions, the Hessian matrix is invertible with probability going to one. Therefore, directly inverting the Hessian matrix serves as a natural alternative to the node-wise lasso for $\widehat{\bTheta}$. In the following, we study the properties of this alternative estimator. Denote $\widetilde{\bTheta} = \widehat{\bSigma}_{\widehat{\bxi}}^{-1}$ to distinguish it from the node-wise lasso estimator $\widehat{\bTheta}$. Similarly, $\widetilde{\bTheta}_j$ represents the $j$th row of $\widetilde{\bTheta}$. 
	
	Similar to  (\ref{eq:derive_bhat}), we have the following equality using $\widetilde{\bTheta}$:
	\begin{equation} \label{eq:derive_bhat_inv}
	\left[
	\widehat{\bxi} + \left( - \widetilde{\bTheta} \mathbb{P}_n \dot{\bm{\rho}}_{\widehat{\bxi}}   \right) +  \left( -\widetilde{\bTheta} \bm{\Delta} \right) +  \left( \widetilde{\bTheta} \mathbb{P}_n \ddot{\bm{\rho}}_{\widehat{\bxi}} - \bI \right) \left( \widehat{\bxi} - \bxi^0 \right) \right] - \bxi^0 = - \widetilde{\bTheta} \mathbb{P}_n \dot{\bm{\rho}}_{{\bxi}^0}.
	\end{equation}
	With $\widetilde{\bTheta} = \widehat{\bSigma}_{\widehat{\bxi}}^{-1}$, the new term $III_j$ in (\ref{eq:derive_bhat_inv}) equals 0 for all $j$, which is no longer a source of bias compared to the original de-biased lasso. Then (\ref{eq:derive_bhat_inv}) becomes 
	\begin{equation} \label{eq:derive_bhat_inv2}
	\left[
	\widehat{\bxi} + \left( - \widetilde{\bTheta} \mathbb{P}_n \dot{\bm{\rho}}_{\widehat{\bxi}}   \right) +  \left( -\widetilde{\bTheta} \bm{\Delta} \right)  \right] - \bxi^0 = - \widetilde{\bTheta} \mathbb{P}_n \dot{\bm{\rho}}_{{\bxi}^0}.
	\end{equation}
	The new de-biased lasso estimator based on  $\widetilde{\bTheta}$ is 
	\[\widetilde{\bm{b}} \equiv \widehat{\bxi} - \widetilde{\bTheta} \mathbb{P}_n \dot{\bm{\rho}}_{\widehat{\bxi}},
	\]  
	which is designed to further correct biases compared to the original de-biased estimator. We will show that any linear combinations of $\widetilde{\bb}$, including each coefficient estimate as a special case, are asymptotically normally distributed. 
	
	\subsection{Theoretical results}
	
	Without loss of generality, we assume that each covariate has been  standardized to have mean zero and variance 1. Let $s_0$ denote the number of non-zero elements in $\bxi^0$. Let $\bX_{\bxi} = \bW_{\bxi}\bX$ be the weighted design matrix, where $\bW_{\bxi}$ is a diagonal matrix with elements $\omega_i(\bxi) = \sqrt{\ddot{\rho}(y_i, x_i^T \bxi)},~ i = 1,\cdots, n$. Recall that for any $\bxi \in \mathbb{R}^{p+1}$,  $\widehat{\bSigma}_{\bxi} = \bX_{\bxi}^T \bX_{\bxi}/n$ and ${\bSigma}_{\bxi} = \E(\widehat{\bSigma}_{\bxi})$. The $\psi_2$-norms (see \citealt{vershynin2010introduction}) introduced below are useful for characterizing the convergence rate of $\widehat{\bSigma}_{\widehat{\bxi}}^{-1}$. 
	For a random variable $Z$, its $\psi_2$-norm  is defined as 
	\[
	\| Z \|_{\psi_2} =  \sup_{r \ge 1} r^{-1/2} (\E|Z|^r)^{1/r}.
	\] 
	We call $Z$ a sub-Gaussian random variable if $\| Z \|_{\psi_2} \le M < \infty$ for a constant $M>0$. For a random vector $\bm{Z}$, its $\psi_2$-norm is defined as $$\|\bm{Z}\|_{\psi_2} =  \sup_{\|\bm{a}\|_2=1} \|\langle \bm{Z}, \bm{a} \rangle\|_{\psi_2}.$$ A random vector $\bm{Z} \in \mathbb{R}^{p+1}$ is called sub-Gaussian if the inner product $\langle\bm{Z}, \ba  \rangle$ is sub-Gaussian for all $\ba \in \mathbb{R}^{p+1}$.   Let $L_p = || \bSigma_{\bxi^0}^{-\frac{1}{2}} \bx_1 \omega_1(\bxi^0) ||_{\psi_2}$, which characterizes the probabilistic tail behavior of the weighted covariates.  We make the following assumptions.
	
	\begin{itemize} \setlength{\itemindent}{15pt}
		\item[(C1)] The elements in $\bX$ are bounded, i.e. there exists a constant $K>0$ such that $\|\bX\|_{\infty} \le K$. 
		\item[(C2)] $\bSigma_{\bxi^0}$ is positive definite and its eigenvalues are bounded and bounded away from $0$, i.e. there exist two absolute constants $c_{\mathrm{min}}$ and $c_{\mathrm{max}} $ such that $0 < c_{\mathrm{min}} \le \lambdamin(\bSigma_{\bxi^0}) \le \lambdamax(\bSigma_{\bxi^0}) \le c_{\mathrm{max}} < \infty$.  
		\item[(C3)] The derivatives $\dot{\rho}(y, a) \equiv \displaystyle \frac{\partial}{\partial a} \rho(y, a)$ and $\ddot{\rho}(y, a) = \displaystyle \frac{\partial^2}{\partial a^2} \rho(y, a)$ exist for all $(y,a)$. For some $\delta$-neighborhood ($\delta > 0$), $\ddot{\rho}(y,a)$ is Lipschitz such that for some absolute constant $c_{Lip} > 0$,
		\[
		\displaystyle \max_{a_0 \in \{\bx_i^T \bxi^0\}}   \sup_{|a-a_0| \vee |\widehat{a} - a_0| \le \delta}  \sup_{y \in \mathcal{Y}}  \displaystyle \frac{| \ddot{\rho}(y,a) -  \ddot{\rho}(y,\widehat{a})|}{|a-\widehat{a}|}	\le c_{Lip}.
		\]
		The derivatives are  bounded in the sense that there exist two constants $K_1, K_2 > 0$ such that
		\[
		\begin{array}{c}
		\displaystyle \max_{a_0 \in \{\bx_i^T \bxi^0\}} \sup_{y \in \mathcal{Y}} |\dot{\rho}(y, a_0)| \le K_1, \\
		\displaystyle \max_{a_0 \in \{\bx_i^T \bxi^0\}} \sup_{|a-a_0|\le \delta} \sup_{y \in \mathcal{Y}} |\ddot{\rho}(y,a)| \le K_2. 
		\end{array}
		\]
		\item[(C4)] $\| \bX \bxi^0 \|_{\infty}$ is bounded.
		\item[(C5)] The matrix $\E({\bX}^T \bX / n)$ is positive definite and its eigenvalues are bounded and bounded away from 0.
	\end{itemize}

	It is common to assume bounded covariates as in (C1) and bounded eigenvalues of the information matrix as in (C2) in high-dimensional inference literature \citep{van2014asymptotically,ning2017general}. $\bx_1, \cdots, \bx_n$ are  sub-Gaussian random vectors under (C1), but we do not impose a boundedness assumption on their $\psi_2$-norm, which may depend on $p$ \citep{vershynin2010introduction,vershynin2012close}. (C2) refers to a compatibility condition that is sufficient to derive the rate of convergence for $\widehat{\bxi}$. (C3) assumes local properties of the derivatives of the general loss $\rho(y, \bx^T \xi)$ \citep{van2014asymptotically}. (C4) is  commonly assumed \citep{van2014asymptotically,ning2017general} and ensures the quadratic margin behavior of the excess risk and is useful to obtain the rate for $\| \bX (\widehat{\bxi} - \bxi^0) \|_2^2 /n$ \citep{buhlmann2011statistics}. (C5) is a  mild requirement in high-dimensional regression analysis with random designs. A similar condition can be found in \citet{wang2011gee}.

	Theorem \ref{thm:main} establishes the asymptotic normality result for any linear combinations of  $\widetilde{\bb}$, based on which inference can be drawn. The proof is given in the Appendix, as well as useful lemmas.

	\medskip
	
	\begin{theorem} \label{thm:main}
		Assume that $\displaystyle L_p^4 {\frac{p^2 \log{p}}{n}} \rightarrow 0$, $\sqrt{p \log(p)}  s_0 \lambda \rightarrow 0$, and $\sqrt{np} s_0\lambda^2 \rightarrow 0$ as $n \rightarrow \infty$. Let  $\widetilde{\bb} = \widehat{\bxi} - \widetilde{\bTheta} \mathbb{P}_n \dot{\rho}_{\widehat{\bxi}}$ and  $\balpha_n \in \mathbb{R}^{p+1}$ with $||\balpha_n||_2 = 1$. Under  (C1) - (C5), we have
		\[
		\displaystyle \frac{\sqrt{n}\balpha_n^T(\widetilde{\bb} - \bxi^0)}{\sqrt{\balpha_n^T \widetilde{\bTheta} \balpha_n}} \overset{d}{\rightarrow} N(0,1).
		\]
	\end{theorem}
	
	\medskip
	
	From Theorem \ref{thm:main}, one can construct $100 \times (1-r)$th confidence intervals for $\balpha_n^T \bxi^0$ as
	\[
	\left[  \balpha_n^T  \bxi^0 - z_{r/2} \sqrt{\balpha_n^T \widetilde{\bTheta} \balpha_n / n}  ,  \balpha_n^T  \bxi^0 + z_{r/2} \sqrt{\balpha_n^T \widetilde{\bTheta} \balpha_n / n}   \right],
	\]
	where $z_{r/2}$ is the upper $(r/2)$th quantile of the standard normal distribution.

	\begin{remark} \normalfont
		For the lasso approach, $\lambda \asymp \sqrt{\log(p)/n}$, we then only need $\displaystyle L_p^2 \sqrt{\frac{p^2 \log{p}}{n}} \rightarrow 0$ and $\sqrt{np} s_0\lambda^2 \rightarrow 0$ as $n \rightarrow \infty$, because $\sqrt{p \log(p)}  s_0 \lambda \rightarrow 0$ and $\sqrt{np} s_0\lambda^2 \rightarrow 0$ are equivalent. 
	\end{remark}
	
	\begin{remark}
		\normalfont
		Theorem 1 reveals that the required rate for $p$ relative to $n$ depends on the factor $L_p$ and can be further simplified. The dependence on  $L_p$ results from that the convergence rate of $\widetilde{\bTheta}$ is related to $L_p = \| \bSigma_{\bxi^0}^{-\frac{1}{2}}\bx_1 \omega_1(\bxi^0) \|_{\psi_2}$. In \citet{javanmard2014confidence} for linear models and \citet{ning2017general} for GLMs, $L_p$ is assumed to be a constant irrelevant to $p$. When covariates follow a  multivariate Gaussian distribution in a linear model,  $L_p = \mathcal{O}(1)$ holds, then it only requires that $\displaystyle \frac{p^2 \log{p}}{n} \rightarrow 0$. However, in general, $L_p$ may grow with $p$, and it can be shown that the utmost bound $L_p = \mathcal{O}(\sqrt{p})$.
		Specifically, by definition, $L_p = \| \bSigma_{\bxi^0}^{-\frac{1}{2}} \bx_1 \omega_1(\bxi^0) \|_{\psi_2} = \sup_{\bz \in B^{p+1}} \| \langle \bSigma_{\bxi^0}^{-\frac{1}{2}} \bx_1 w_1(\bxi^0), \bz \rangle \|_{\psi_2}$, where $B^{p+1}$ is the unit ball in $\mathbb{R}^{p+1}$. Then we have
		\begin{eqnarray*}
			| \langle \bSigma_{\bxi^0}^{-\frac{1}{2}} \bx_1 w_1(\bxi^0), \bz \rangle | 
			&\le & \|\bz\|_2 \cdot \| \bSigma_{\bxi^0}^{-\frac{1}{2}} \bx_1 w_1(\bxi^0) \|_2 \\
			&\le & \| \bSigma_{\bxi^0}^{-\frac{1}{2}} \| \cdot \| \bx_1 w_1(\bxi^0) \|_2 \\
			&\le & c_{\min}^{-\frac{1}{2}} \sqrt{K_2 ({p+1})}K.
		\end{eqnarray*}
		Therefore, $L_p \le c_{\min}^{-\frac{1}{2}} \sqrt{K_2 ({p+1})}K$. 
		This results in the most stringent rate requirement  
		$\displaystyle \frac{p^4 \log{p}}{n} \rightarrow 0$,  implying  $\sqrt{np} s_0 \lambda^2 = o(1)$ when $\lambda \asymp \sqrt{\log(p)/n}$.

	\end{remark}
	
	\begin{remark}
		\normalfont
		In Theorem 1, $p$ is assumed to grow slowly  with  $n$ so that $p \ll n$. 
		This assumption is not uncommon in the literature. \citet{fan2004nonconcave} assumed $p^5 / n \rightarrow 0$ for a non-concave penalized maximum likelihood estimator to establish the oracle property and the asymptotic normality for selected variables. Yet the estimates in \citet{fan2004nonconcave} are super-efficient, which is not our focus. Without parameter regularization, \citet{wang2011gee} assumed $p^3 / n \rightarrow 0$ to derive asymptotic normality for the solutions to generalized estimating equations with binary outcomes and clustered data, which reduces to the usual logistic regression when simplified to a singleton in each cluster. \citet{wang2011gee} studied a fixed design case, and proved the asymptotic normality for a different quantity $\balpha_n^T \overline{\mathbf{M}}_n(\bbeta_{n0})^{-1/2} \overline{\mathbf{H}}_n (\bbeta_{n0})(\widehat{\bbeta}_n - \bbeta_{n0})$; see  Theorem 3.8 in \citet{wang2011gee}. When $p/n$ is not negligible (e.g. $>0.1$), simulations show that  MLE yields biased and highly variable estimates, and is outperformed by our proposed $\widetilde{b}$. 
		
	\end{remark}

	\subsection{Simulation results}

	We investigate the performance of our alternative de-biased estimator $\widetilde{\bb}$ in the ``large $n$, diverging $p$" scenario, and focus on   biases in estimates and coverage probabilities of confidence intervals. The estimators in comparison are
	\begin{itemize} \setlength{\itemsep}{0pt}
		\item[(i)] the original de-biased lasso estimator $\widehat{b}_j$ obtained by using the node-wise lasso estimator $\widehat{\bTheta}$ in \citet{van2014asymptotically} (\textit{ORIG-DS});
		\item[(ii)] the refined de-biased lasso approach based on the inverse matrix estimation  $\widetilde{\bTheta} = \widehat{\bSigma}_{\widehat{\bxi}}^{-1}$, $\widetilde{b}_j$, as described in this section (\textit{REF-DS});
		\item[(iii)] the conventional MLE (\textit{MLE}).
	\end{itemize}

	As simulations using  logistic and Poisson regression models yield  similar results, we only report those from logistic regression.  A total of $n=1,000$ observations and $p = 40, 100, 300, 400$ covariates are simulated. We assume that in $\bx_i = (1, \widetilde{\bx}_i^T)^T$, $\widetilde{\bx}_i$ are independently generated from $N_{p} (\mathbf{0}_{p}, \bSigma_x)$ then truncated at $\pm 6$, and $y_i | \bx_i \sim Bernoulli(\mu_i)$, where $\mu_i \equiv \exp(\bx_i^T \bxi^0)/\{ 1+ \exp(\bx_i^T \bxi^0) \}$. The intercept $\beta_0^0 = 0$, and  $\beta_1^0$  varies from 0 to 1.5 with 40 equally spaced increments. Four additional arbitrarily chosen elements of $\bbeta^0$ take non-zero values, two at $0.5$ and the other two at 1, and then are fixed throughout the simulation. In some settings, \textit{MLE} estimates do not exist due to divergence and thus are not shown. The covariance matrix $\bSigma_x$ for $\widetilde{\bx}_i$ takes one of the following three forms: identity matrix, AR(1) with correlation $\rho = 0.7$, and compound symmetry with correlation $\rho=0.7$. The tuning parameter in the $\ell_1$-norm penalized regression is selected by 10-fold cross-validation, and the tuning parameter for the node-wise lasso estimator $\widehat{\bTheta}$ is selected using 5-fold cross-validation. Both tuning parameter selection procedures are implemented using  \texttt{glmnet} \citep{friedman2010regularization}. For every $\beta_1^0$ value, we summarize the average bias, empirical coverage probability, empirical standard error and model-based estimated standard error over 200 replications.

	Figure \ref{fig:logit_n1k_ar1} presents the simulation results for estimating $\beta_1^0$ under the AR(1) covariance structure, whereas the similar simulation results under the other two covariance structures are provided in the online supplementary material. The three methods in comparison behave similarly when only 40 covariates are present, with \textit{MLE} showing slightly larger biases for larger signals. \textit{MLE} displays much more biases than the other two methods when 100 covariates are present, and does not always exist in some settings as the number of covariates increases. When  \textit{MLE} does exist, it shows more variability than \textit{ORIG-DS} and \textit{REF-DS}, and  lower coverage probabilities. There is a systematic bias in \textit{ORIG-DS}, which increases with the signal strength of $\beta_1^0$. For large signals, the model-based standard error of \textit{ORIG-DS} slightly underestimates the true variability. These factors contribute to the poor coverage probabilities of \textit{ORIG-DS} when signal size is not too close to zero. Among all the competing methods, \textit{REF-DS} presents the least biases and has an empirical coverage probability closest to the nominal level across different settings, though \textit{REF-DS} exhibits slightly higher variability than \textit{ORIG-DS}.
	This is possibly because \textit{REF-DS} does not utilize penalization when inverting the matrix. Under the null $\beta_1^0 = 0$, both \textit{ORIG-DS} and \textit{REF-DS} have coverage probabilities close to 95\% and preserve the type 1 error.

	\subsection{Boston Lung Cancer Study (BLCS)}
	
	Lung cancer is the leading cause of cancer death in the United States. BLCS, a large epidemiology cohort for investigating the molecular cause of lung cancer,  includes over 11,000 lung cancer cases  enrolled at Massachusetts General Hospital and the Dana-Farber Cancer Institute from 1992 to present (see  \url{https://maps.cancer.gov/overview/DCCPSGrants/abstract.jsp?applId=9320074&term=CA209414}). We applied \textit{REF-DS}, together with \textit{ORIG-DS} and \textit{MLE}, to a subset of the BLCS data and simultaneously examined the joint effects of SNPs in nine target genes on the overall risk of lung cancer. 
	
	Genotypes from Axiom array and clinical information were originally collected on 1,459 individuals. Out of the 1,459 individuals, 14 (0.96\%) had missing smoking status, 8 (0.55\%) had missing race information, and 1,386 (95\%) were Caucasian. We included a final number of $n=1,374$ Caucasians, where $n_0=723$ were controls and $n_1=651$ were cases, with known lung cancer status (``1" for cases and ``0" for controls) and smoking status (``1" for ever smoker and ``0" for never). Among the 1,077 smokers, 595 had lung cancer, and the number of cases was 56 out of the 297 non-smokers. Other demographic characteristics of the study population, including education level (no high school, high school graduate, or at least 1-2 years of college), gender and age, are summarized in the online supplementary material. Using the target gene approach, we focused on the following genes: \textit{AK5} on region 1p31.1, \textit{RNASET2} on region 6q27,  \textit{CHRNA2} and \textit{EPHX2} on region 8p21.2, \textit{BRCA2} on region 13q13.1,  \textit{SEMA6D} and \textit{SECISBP2L} on region 15q21.1,  \textit{CHRNA5} on region 15q25.1, and \textit{CYP2A6} on region 19q13.2. These genes have been reported in  \citet{mckay2017large} to harbor SNPs  associated  with the overall lung cancer risks. In our dataset, each SNP was coded as 0,1,2, reflecting the number of copies of the minor allele, and was assumed  to have ``additive effects". After applying filters on the minor allele frequency,  genotype call rate (percentage of missingness), and excluding SNPs that were highly correlated, 103 SNPs remained. Details on  data processing can be found in the online supplementary material.

	The final analyzable dataset consisted of 1,374 individuals, 103 SNPs, and demographic information including education history, age and gender. Since existing studies suggest smoking can modify   associations  between lung cancer risks and SNPs, for example, those residing in region 15q25.1  \citep{gabrielsen2013association,amos2008genome}, we conducted analysis stratified by smoking status. Within the smoker and non-smoker groups, we fitted separate logistic regression models, adjusting for educational history, gender and age (centered at the mean). In total, there were 107 variables for stratified analysis among 1,077 smokers and 297 non-smokers. As  a reference, we conducted  marginal analysis, which examined one SNP at a time while adjusting for demographic information. Marginal and joint analyses have distinct interpretations and can generate different estimates.

	We applied  these methods to  draw inference on all of the 107 predictors, and comparisons of the results of the BLCS data analysis may shed light on the molecular mechanism underlying lung cancer. For ease of presentation, Table \ref{tab:blcs_strat1} lists the regression coefficient estimates, model-based estimated standard errors and  95\% confidence intervals (CIs) for demographic variables and 11 SNPs in the stratified analysis for an illustration. Some of these SNPs had at least one 95\% CI (calculated by the three methods)  that excluded 0 among either the smokers or the non-smokers; others showed differences among the estimating methods. Details of the remaining SNPs were omitted due to the space limitation.
	Since the number of the non-smokers was only about one third of the smokers, the \textit{MLE} estimates had the largest standard errors and tended to break down among the non-smokers (see, for example, AX-62479186 in Table \ref{tab:blcs_strat1}\subref{tab:blcs_nonsmo}), whereas the two de-biased lasso methods gave more reasonable estimates. The estimates by \textit{REF-DS} and \textit{ORIG-DS} shared more similarity in the smokers (Table \ref{tab:blcs_strat1}\subref{tab:blcs_smo}) than in the non-smokers (Table \ref{tab:blcs_strat1}\subref{tab:blcs_nonsmo}).  Overall, \textit{ORIG-DS} had slightly narrower confidence intervals than \textit{REF-DS}, probably due to penalized estimation for $\widehat{\bTheta}$.   These results generally agreed with our  simulation results. 
	
	Additional differences between \textit{ORIG-DS} and \textit{REF-DS} lied in opposite directions obtained for the estimated effects of some SNPs, such as  AX-38419741 and AX-15934253 in Table \ref{tab:blcs_strat1}\subref{tab:blcs_smo}, and AX-42391645 in Table \ref{tab:blcs_strat1}\subref{tab:blcs_nonsmo}. Among the non-smokers, the 95\% CI for AX-31620127 in \textit{SEMA6D} by \textit{REF-DS} was all positive and excluded 0, while the CI by \textit{ORIG-DS} included 0; the story for AX-88907114 in \textit{CYP2A6} was just opposite (Table \ref{tab:blcs_strat1}\subref{tab:blcs_nonsmo}).

	\textit{CHRNA5} is a gene  known  for predisposition to nicotine dependence \citep{hallden2016gene,hung2008susceptibility,amos2008genome,thorgeirsson2008variant,gabrielsen2013association}. Though AX-39952685 and  AX-88891100 in \textit{CHRNA5} were not significant at level 0.05 in marginal analysis among the smokers, their 95\% CIs in Table \ref{tab:blcs_strat1}\subref{tab:blcs_smo} excluded 0 by all of the three methods.  Indeed AX-88891100, or rs503464 mapped to the same physical location in dbSNP (\url{https://www.ncbi.nlm.nih.gov/snp/}), was found to ``decrease \textit{CHRNA5} promoter-derived luciferase activity" \citep{doyle2011vitro}. The same SNP was also reported to be significantly associated with nicotine dependence at baseline, as well as response to varenicline, bupropion, nicotine replacement therapy for smoking cessation \citep{pintarelli2017pharmacogenetic}.  AX-39952685 was found to be strongly correlated with SNP AX-39952697 in \textit{CHRNA5}, which was mapped to the same physical location as rs11633585 in dbSNP. All of these markers  were found to be significantly associated with nicotine dependence \citep{stevens2008nicotinic}. The stratified analysis also suggested molecular mechanisms of lung cancer differ between smokers and non-smokers, but affirmative conclusions need additional confirmatory studies. In summary, jointly modeling the genetic effects on lung cancer risks can help understand underlying mechanisms and personalized therapies, which necessitates the use of reliable inference tools.

	\section{Discussion} 
	
	Our work has produced several intriguing results that can be impactful in both theory and practical implementation. From extensive simulations we have discovered the unsatisfactory performance of de-biased lasso in drawing inference with high-dimensional GLMs. We have further
	pinpointed an essential assumption that hardly holds for GLMs in general, i.e. the sparsity of the high-dimensional inverse information matrix $\bTheta_{\bxi^0}$ \citep{van2014asymptotically}, making  de-biased lasso fail to deliver reliable inference in practice. This type of $\ell_0$ sparsity conditions on matrices is not uncommon in the literature of high-dimensional inference. A related $\ell_0$ sparsity condition on $\mathbf{w}^* = {\bI^{* -1}}_{\mathbf{\gamma}\mathbf{\gamma}} \bI^*_{\mathbf{\gamma} \theta}$ can be found in \citet{ning2017general}, where $\bI^*$ is the information matrix under the truth, but is not well justified in a general GLM setting. When testing a global null hypothesis ($\bbeta^0= \mathbf{0}$), however, the sparsity of $\bTheta_{\bxi^0}$ reduces to the sparsity of the covariate precision matrix, which becomes less of an issue (see \citealt{cai2019sparse}).

	Our detailed work leads to practical guidelines as to how to use de-biased lasso for proper statistical inference with high-dimensional GLMs.  Our work summarily suggests that, when $p>n$,  de-biased lasso may not be applicable in general; when $p<n$ with diverging $p$, it is preferred to use the refined de-biased lasso, which directly inverts the Hessian matrix and provides improved confidence interval coverage probabilities for a wide range of $p$; when $p$ is rather small relative to $n$ (often viewed as a fixed $p$ problem), the refined de-biased lasso yields results nearly identical to MLE and the original de-biased lasso.


	\clearpage
	\appendix
	
	\begin{center}
		\Large \textbf{Appendix: Lemmas and Proofs}
	\end{center}

	We provide three lemmas that are useful for proving Theorem 1. 
	Without loss of generality, we denote the dimension of the parameter $\bxi$ by $p$ instead of $(p+1)$ to simplify the notation in the  proofs. Consequently, the matrices such as $\bSigma_{\bxi}$ and $\bTheta_{\bxi}$ are considered as $p\times p$ matrices. The simplification of notation does not affect  derivations.

	\begin{lemma}
		Under  (C1) - (C4), we have $\| \widehat{\bxi} - \bxi^0 \|_1 = \OP(s_0\lambda)$ and $\| \bX (\widehat{\bxi} - \bxi^0) \|_2^2 /n = \OP(s_0 \lambda^2)$.
	\end{lemma}

	\begin{proof}[\textbf{Proof}]
		
		Because $\lambda_{\mathrm{min}}(\bSigma_{\bxi^0}) > 0$ in (C2), the compatibility condition holds for all index sets $S \subset \{ 1, \cdots, p\}$ by Lemma 6.23 \citep{buhlmann2011statistics} and the fact that the adaptive restricted eigenvalue condition implies the compatibility condition. Exploiting  Hoeffding's concentration inequality, we have $\| \widehat{\bSigma}_{\bxi^0} - \bSigma_{\bxi^0} \|_{\infty} = \OP(\sqrt{log(p)/n})$. Then by Lemma 6.17 of \citet{buhlmann2011statistics}, we have the $\widehat{\bSigma}_{\bxi^0}$-compatibility condition. Finally, the first part of Lemma 1 follows from Theorem 6.4 in \citet{buhlmann2011statistics}.  
		
		For the second claim, \citet{ning2017general} showed that $(\widehat{\bxi} - \bxi^0)^T \widehat{\bSigma}_{\bxi^0} (\widehat{\bxi} - \bxi^0)^T = \OP(s_0 \lambda^2)$, then under (C4), we obtain the desired result. 
	\end{proof}

	\begin{lemma}
		Under (C1) - (C5), if we further assume that $s_0 \lambda \rightarrow 0$ and  $L_p^2 \sqrt{\displaystyle \frac{p}{n}} \rightarrow 0$, then $\widetilde{\bTheta}$ converges with the following rate
		\[
		|| \widetilde{\bTheta} - \bTheta_{\bxi^0} ||= \OP \left( L_p^2 \sqrt{\displaystyle \frac{p}{n}} + s_0 \lambda  \right).
		\]
	\end{lemma}

	\begin{proof}[\textbf{Proof}]
		Since $\widehat{\bSigma}_{\widehat{\bxi}}^{-1} - \bSigma_{\bxi^0}^{-1} = \widehat{\bSigma}_{\widehat{\bxi}}^{-1} \left(\bSigma_{\bxi^0} -  \widehat{\bSigma}_{\widehat{\bxi}}  \right) \bSigma_{\bxi^0}^{-1}$, we have 
		\begin{equation} \label{eq:inv}
		\| \widehat{\bSigma}_{\widehat{\bxi}}^{-1} - \bSigma_{\bxi^0}^{-1} \| \le  \| \widehat{\bSigma}_{\widehat{\bxi}}^{-1}  \| \cdot \| \widehat{\bSigma}_{\widehat{\bxi}}  - \bSigma_{\bxi^0}  \|  \cdot \| \bSigma_{\bxi^0}^{-1} \|.
		\end{equation}
		By (C2), $\| \bSigma_{\bxi^0}^{-1} \|$ is bounded. We obtain the convergence rate of $\| \widehat{\bSigma}_{\widehat{\bxi}}^{-1} - \bSigma_{\bxi^0}^{-1} \|$ by calculating the rate of $\| \widehat{\bSigma}_{\widehat{\bxi}} - \bSigma_{\bxi^0}  \|$ and showing that $\| \widehat{\bSigma}_{\widehat{\bxi}}^{-1}  \| $ is bounded with probability going to 1.
		
		Note that $\| \widehat{\bSigma}_{\widehat{\bxi}} - \bSigma_{\bxi^0}  \| \le \| \widehat{\bSigma}_{\widehat{\bxi}} - \widehat{\bSigma}_{\bxi^0} \| + \| \widehat{\bSigma}_{\bxi^0} -  \bSigma_{\bxi^0}  \|$. When the rows of $\bX$ are sub-Gaussian, so are the rows of $\bX_{\bxi^0}$ due to the boundedness of the weights $w_i$ in (C3). First, for $\| \widehat{\bSigma}_{\bxi^0} -  \bSigma_{\bxi^0}  \|$, \citet{vershynin2010introduction} shows that for every $t > 0$, it holds with probability at least  $1 - 2\exp(-c_L^{\prime} t^2)$ that
		\begin{equation}
		\| \widehat{\bSigma}_{\bxi^0} -  \bSigma_{\bxi^0} \| \le \| \bSigma_{\bxi^0}  \| \max(\delta, \delta^2) \le c_{\mathrm{max}} \max(\delta, \delta^2),
		\end{equation}
		where $\delta = C_L \sqrt{\displaystyle \frac{p}{n}} + \displaystyle \frac{t}{\sqrt{n}}$. Here $C_L$, $c_L^{\prime} > 0$ depend only on  $L_p = \| \bSigma_{\bxi^0}^{-\frac{1}{2}} \bx_1 \omega_1(\bxi^0) \|_{\psi_2}$. In fact $c_L^{\prime} = c_1/L_p^4$ and $C_L = L_p^2 \sqrt{\log 9 / c_1}$, where $c_1$ is an absolute constant. For $s > 0$ and $t = s C_L \sqrt{p}$, the probability becomes $1 - 2\exp(-c_2 s^2 p)$, $c_2 > 0$ being some absolute constant, and $\delta = (s+1) C_L \displaystyle \sqrt{\frac{p}{n}}$. Thus $\| \widehat{\bSigma}_{\bxi^0} -  \bSigma_{\bxi^0} \| = \mathcal{O}_p \left( \displaystyle  L_p^2 \sqrt{\frac{p}{n}} \right)$. 
		
		Note that 
		\[
		\begin{array}{rcl}
		\| \widehat{\bSigma}_{\widehat{\bxi}} - \widehat{\bSigma}_{\bxi^0} \| & = & \| {\bX}^T (\bW^2_{\widehat{\bxi}} - \bW^2_{\bxi^0}) \bX / n \|  \\ 
		& \le & \| {\bX}^T \| \cdot \| \bX \| / n \cdot \| \bW^2_{\widehat{\bxi}} - \bW^2_{\bxi^0} \|  \\
		& = & \lambdamax({\bX}^T \bX /n) \cdot  \| \bW^2_{\widehat{\bxi}} - \bW^2_{\bxi^0} \|.
		\end{array}
		\]
		By (C1) and  (C3), \begin{equation}
		\begin{array}{rcl}
		\| \bW^2_{\widehat{\bxi}} - \bW^2_{\bxi^0} \| & = & \max_i | \ddot{\rho}(y_i , \bx_i^T \widehat{\bxi}) - \ddot{\rho}(y_i , \bx_i^T \bxi^0)  | \\
		& \le &  c_{Lip}  \cdot \max_i | \bx_i^T (\widehat{\bxi} - \bxi^0)| \\
		& \le &  c_{Lip} K  \cdot \| \widehat{\bxi} - \bxi^0 \|_1.
		\end{array}
		\end{equation}
		By Lemma 1, we have $\| \widehat{\bxi} - \bxi^0 \|_1 = \OP(s_0 \lambda)$. In this case, $\| \bW^2_{\widehat{\bxi}} - \bW^2_{\bxi^0} \| = \OP(s_0 \lambda)$. By (C5) and \citet{vershynin2010introduction}, $\lambdamax({\bX}^T \bX /n) = \OP(1)$. Thus $\| \widehat{\bSigma}_{\widehat{\bxi}} - \widehat{\bSigma}_{\bxi^0} \| = \OP( s_0 \lambda)$. 
		Therefore, after combining the two parts, we have $ \| \widehat{\bSigma}_{\widehat{\bxi}} - \bSigma_{\bxi^0} \| = \OP \left( L_p^2 \sqrt{\displaystyle \frac{p}{n}} +  s_0 \lambda \right)$. 
		Under $L_p^2 \displaystyle \sqrt{\frac{p}{n}} = o(1)$ and $s_0 \lambda = o(1)$, we have $ \| \widehat{\bSigma}_{\widehat{\bxi}} - \bSigma_{\bxi^0} \| = \oP(1)$.
		
		Now for any vector $\bx$ with $\|\bx\|_2 = 1$, we have
		\[
		\displaystyle \inf_{\|\bm{y}\|_2 = 1} \| \widehat{\bSigma}_{\widehat{\bxi}} \bm{y}  \|_2 \le \| \widehat{\bSigma}_{\widehat{\bxi}} \bx \|_2 \le \| \bSigma_{\bxi^0} \bx \|_2 + \| (\widehat{\bSigma}_{\widehat{\bxi}} - \bSigma_{\bxi^0} ) \bx \|_2 \le \| \bSigma_{\bxi^0} \bx \|_2 + \displaystyle \sup_{\|\bm{z}\|_2 = 1} \| (\widehat{\bSigma}_{\widehat{\bxi}} - \bSigma_{\bxi^0} ) \bm{z} \|_2,
		\]
		which indicates that $\lambdamin(\widehat{\bSigma}_{\widehat{\bxi}}) \le \lambdamin(\bSigma_{\bxi^0}) + \| \widehat{\bSigma}_{\widehat{\bxi}} - \bSigma_{\bxi^0} \|$. Similarly, we have $\lambdamin(\bSigma_{\bxi^0})  \le  \lambdamin(\widehat{\bSigma}_{\widehat{\bxi}}) + \| \widehat{\bSigma}_{\widehat{\bxi}} - \bSigma_{\bxi^0} \|$. So $ | \lambdamin(\bSigma_{\bxi^0}) - \lambdamin(\widehat{\bSigma}_{\widehat{\bxi}}) | \le \| \widehat{\bSigma}_{\widehat{\bxi}} - \bSigma_{\bxi^0} \|$. For any $0 < \epsilon < \min \{ \| \bSigma_{\bxi^0} \|, \lambdamin( \bSigma_{\bxi^0} )/2 \}$, we have that
		\begin{eqnarray*}
			P \left( \| \widehat{\bSigma}_{\widehat{\bxi}}^{-1} \| \ge \displaystyle \frac{1}{\lambdamin(\bSigma_{\bxi^0}) - \epsilon} \right) 
			&= & P (  \lambdamin(\widehat{\bSigma}_{\widehat{\bxi}}) \le \lambdamin(\bSigma_{\bxi^0}) - \epsilon  )  \\
			&\le & P ( | \lambdamin(\widehat{\bSigma}_{\widehat{\bxi}}) - \lambdamin(\bSigma_{\bxi^0})  | \ge \epsilon ) \\
			&\le & P ( \| \widehat{\bSigma}_{\widehat{\bxi}} - \bSigma_{\bxi^0} \| \ge \epsilon ).
		\end{eqnarray*}
		Since $ \| \widehat{\bSigma}_{\widehat{\bxi}} - \bSigma_{\bxi^0} \| = \oP(1)$, we have $\| \widehat{\bSigma}_{\widehat{\bxi}}^{-1} \| = \OP(1)$. Finally, by (\ref{eq:inv}), $ \| \widehat{\bSigma}_{\widehat{\bxi}}^{-1} - \bSigma_{\bxi^0}^{-1} \| =  \OP( \| \widehat{\bSigma}_{\widehat{\bxi}} - \bSigma_{\bxi^0} \|) = \OP \left( L_p^2 \sqrt{\displaystyle \frac{p}{n}} +  s_0 \lambda \right)$.
	\end{proof}

	\begin{lemma}
		Under  (C1)-(C3), when $\displaystyle \frac{p}{n} \rightarrow 0$, it holds that for any vector $\balpha_n \in \mathbb{R}^{p}$ with $\|\balpha_n\|_2 = 1$,
		\[
		\displaystyle \frac{\sqrt{n} \balpha_n^T \bTheta_{\bxi^0} \mathbb{P}_n \dot{\bm{\rho}}_{\bxi^0}}{\sqrt{\balpha_n^T\bTheta_{\bxi^0} \balpha_n}} \overset{d}{\rightarrow} N(0,1).
		\]
	\end{lemma}

	\begin{proof}[\textbf{Proof}]
		We invoke the Lindeberg-Feller Central Limit Theorem. For $i = 1, \cdots, n$, let 
		\[
		Z_{ni} = \displaystyle \frac{1}{\sqrt{n}} \balpha_n^T \bTheta_{\bxi^0} \dot{\bm{\rho}}_{\bxi^0}(y_i, \bx_i) = \displaystyle \frac{1}{\sqrt{n}} \balpha_n^T \bTheta_{\bxi^0} \bx_i \dot{\rho}(y_i, \bx_i^T \bxi^0),
		\]
		and $s_n^2 = Var \left( \sum_{i=1}^n Z_{ni} \right)$. Note that $\E[\dot{\rho}(y_i, \bx_i^T \bxi^0) | \bx_i]=0$ and consequently $\E(Z_{ni})=0$. Because $\{ (y_i, \widetilde{\bx}_i)\}_{i=1}^n$ are $i.i.d.$, we can show that $s_n^2 = \balpha_n^T \bTheta_{\bxi^0} \balpha_n$.
		To show $\displaystyle \frac{\sum_{i=1}^n Z_{ni}}{s_n} \overset{d}{\rightarrow} N(0,1)$, we first check the Lindeberg condition and then the conclusion shall follow by the Lindeberg-Feller Central Limit Theorem. Specifically, for any $\epsilon > 0$, we show that as $n \rightarrow \infty$,
		\[
		\displaystyle \frac{1}{s_n^2} \sum_{i=1}^n \E \left\{ Z_{ni}^2 \cdot {1}_{(|Z_{ni}| > \epsilon s_n)} \right\} \rightarrow 0.
		\]
		Due to the boundedness of the eigenvalues of $\bSigma_{\bxi^0}$, $\balpha_n ^T \bTheta_{\bxi^0} \balpha_n \ge \lambdamin(\bTheta_{\bxi^0}) = 1/\lambdamax(\bSigma_{\bxi^0}) \ge c_{\mathrm{max}}^{-1}$. On the other hand, by the  Cauchy-Schwarz inequality, it holds almost surely that
		\[
		\left(  \balpha_n^T \bTheta_{\bxi^0} \bx_i  \right)^2  \le  \| \balpha_n \|_2^2 \cdot \| \bTheta_{\bxi^0} \bx_i \|_2^2 
		\le  \left[   \| \bTheta_{\bxi^0} \| \cdot \|\bx_i\|_2  \right]^2 
		\le  c_{\mathrm{min}}^{-2} \cdot \mathcal{O}(pK^2).
		\]
		Inside the indicator, it holds almost surely that
		\[
		\begin{array}{rcl}
		\displaystyle \frac{Z_{ni}^2}{s_n^2} & = & \displaystyle \frac{[\dot{\rho}(y_i, \bx_i^T \bxi_0)]^2 \left(  \balpha_n^T \bTheta_{\bxi^0} \bx_i  \right)^2}{n \balpha_n ^T \bTheta_{\bxi^0} \balpha_n} \\
		& \le & [\dot{\rho}(y_i, \bx_i^T \bxi_0)]^2 \cdot c_{\mathrm{min}}^{-2} c_{\mathrm{max}} \cdot \mathcal{O}(K^2 \displaystyle \frac{p}{n}) \\
		& \le & K_1^2 c_{\mathrm{min}}^{-2} c_{\mathrm{max}} \cdot \mathcal{O}(K^2 \displaystyle \frac{p}{n}),
		\end{array}
		\]
		where the last inequality follows from the boundedness of $\dot{\rho}(y_i, \bx_i^T \bxi_0)$ in condition (C3).  Hence, we have $Z_{ni}^2 / s_n^2 \rightarrow 0$ almost surely as $p/n \rightarrow 0$. When $n$ is large enough, $Z_{ni}^2 / s_n^2 < \epsilon^2$ and all the indicators become 0. Therefore, by the Dominated Convergence Theorem, the Lindeberg condition holds and the Lindeber-Feller Central Limit Theorem guarantees the asymptotic normality.  
	\end{proof}

	\begin{proof}[\textbf{Proof of Theorem 1}]
		Recall that from (\ref{eq:derive_bhat_inv2}), 
		\[
		\sqrt{n} \balpha_n^T (\widetilde{\bb} - \bxi^0) - \sqrt{n} \balpha_n^T \widetilde{\bTheta}\bDelta = - \sqrt{n} \balpha_n^T \widetilde{\bTheta} \Pn \dot{\bm{\rho}}_{\bxi^0}.
		\]
		First, we show that $\balpha_n^T \widetilde{\bTheta} \balpha_n - \balpha_n^T {\bTheta}_{\bxi^0} \balpha_n = o_{\mathbb{P}}(1)$ and that $\displaystyle \frac{\sqrt{n}\balpha_n^T \widetilde{\bTheta} \mathbb{P}_n \dot{\bm{\rho}}_{\bxi^0}}{\sqrt{\balpha_n^T \widetilde{\bTheta} \balpha_n}} = \displaystyle \frac{\sqrt{n}\balpha_n^T {\bTheta}_{\bxi^0} \mathbb{P}_n \dot{\bm{\rho}}_{\bxi^0}}{\sqrt{\balpha_n^T {\bTheta}_{\bxi^0} \balpha_n}} + o_{\mathbb{P}}(1)$. Then by Slutsky's Theorem, the asymptotic distribution of the target $\displaystyle \frac{\sqrt{n}\balpha_n^T \widetilde{\bTheta} \mathbb{P}_n \dot{\bm{\rho}}_{\bxi^0}}{\sqrt{\balpha_n^T \widetilde{\bTheta} \balpha_n}}$ can be derived by using the asymptotic distribution of  $\displaystyle \frac{\sqrt{n}\balpha_n^T {\bTheta}_{\bxi^0} \mathbb{P}_n \dot{\bm{\rho}}_{\bxi^0}}{\sqrt{\balpha_n^T {\bTheta}_{\bxi^0} \balpha_n}}$, which has been proved in Lemma 3. In the final step, as long as $\sqrt{n} \balpha_n^T \widetilde{\bTheta} \bDelta = \oP(1)$, the asymptotic distribution of $\displaystyle \frac{\sqrt{n}\balpha_n^T(\widetilde{\bb} - \bxi^0)}{\sqrt{\balpha_n^T \widetilde{\bTheta} \balpha_n}}$ follows immediately.
		
		According to Lemma 2, it follows that 
		\[
		\vert \balpha_n^T \widetilde{\bTheta} \balpha_n - \balpha_n^T {\bTheta}_{\bxi^0} \balpha_n \vert = \vert \balpha_n^T  (\widetilde{\bTheta} - {\bTheta}_{\bxi^0} ) \balpha_n \vert \le \| \widetilde{\bTheta} - {\bTheta}_{\bxi^0} \| \cdot \| \balpha_n \|_2^2 =   \oP(1).
		\]
		By the Cauchy-Schwartz inequality,
		\[
		\sqrt{n} \vert \balpha_n^T \widetilde{\bTheta} \mathbb{P}_n \dot{\bm{\rho}}_{\bxi^0} - \balpha_n^T {\bTheta}_{\bxi^0} \mathbb{P}_n \dot{\bm{\rho}}_{\bxi^0} \vert \le \sqrt{n} \| \balpha_n \|_2 \cdot \| (\widetilde{\bTheta} - \bTheta_{\bxi^0}) \Pn \dot{\bm{\rho}}_{\bxi^0} \|_2.
		\]
		Since 
		\[
		\begin{array}{rcl}
		\| (\widetilde{\bTheta} - \bTheta_{\bxi^0}) \Pn \dot{\bm{\rho}}_{\bxi^0} \|_2 & \le & \| \widetilde{\bTheta} - \bTheta_{\bxi^0} \| \cdot \| \Pn \dot{\bm{\rho}}_{\bxi^0} \|_2 \\
		& \le  &  \| \widetilde{\bTheta} - \bTheta_{\bxi^0}   \| \cdot \sqrt{p}  \|  \Pn \dot{\bm{\rho}}_{\bxi^0} \|_{\infty},
		\end{array}
		\]
		we have 
		\[ \begin{array}{rcl}
		\sqrt{n} \left\vert \balpha_n^T \widetilde{\bTheta} \mathbb{P}_n \dot{\bm{\rho}}_{\bxi^0} - \balpha_n^T {\bTheta}_{\bxi^0} \mathbb{P}_n \dot{\bm{\rho}}_{\bxi^0} \right\vert & \le & \sqrt{np}  \cdot \|  \Pn \dot{\bm{\rho}}_{\bxi^0} \|_{\infty} \cdot \OP\left( \displaystyle  L_p^2\sqrt{ \frac{p}{n}} + s_0 \lambda \right) \\
		& = & \|  \Pn \dot{\bm{\rho}}_{\bxi^0} \|_{\infty} \cdot \OP\left( \displaystyle L_p^2 p + \sqrt{np}  s_0 \lambda \right).
		\end{array}
		\]
		By definition, 
		\[
		\| \Pn \dot{\bm{\rho}}_{\bxi^0} \|_{\infty} = \max_j \left\vert  \displaystyle \frac{1}{n} \sum_{i=1}^n \dot{\bm{\rho}}_{\bxi^0}(y_i, \bx_i) \right\vert = \max_j \left\vert  \displaystyle \frac{1}{n} \sum_{i=1}^n x_{ij}\dot{\rho}(y_i, \bx_i^T \bxi^0) \right\vert.
		\]
		Assume $|\dot{\rho}(y_i, \bx_i^T \bxi^0)| \le K_1$ for all $i$ and the constant $K_1 > 0$ in condition (C3). As $| x_{ij} \dot{\rho}(y_i, \bx_i^T \bxi^0) | \le K K_1$ almost surely holds for all $i$ and $j$, we apply Lemma 14.15 in \citet{buhlmann2011statistics}, for all $t > 0$,
		\[
		\mathbb{P} \left( \max_j \left\vert  \displaystyle \frac{1}{n} \sum_{i=1}^n x_{ij}\dot{\rho}(y_i, \bx_i^T \bxi^0) \right\vert \ge K K_1 \sqrt{2\left( t^2 + \displaystyle \frac{\log(2p)}{n} \right)}   \right) \le \exp[-nt^2].
		\] 
		For $t^2 = \displaystyle \frac{\log(2p)}{n}$, we know that $\| \Pn \dot{\bm{\rho}}_{\bxi^0} \|_{\infty} = \OP \left(  \sqrt{ \displaystyle \frac{\log(p)}{n} } \right)$. Then we have 
		\[
		\sqrt{n} \left\vert \balpha_n^T \widetilde{\bTheta} \mathbb{P}_n \dot{\bm{\rho}}_{\bxi^0} - \balpha_n^T {\bTheta}_{\bxi^0} \mathbb{P}_n \dot{\bm{\rho}}_{\bxi^0} \right\vert \le 
		\OP \left( L_p^2 p \sqrt{\displaystyle \frac{\log(p)}{n}} +  s_0 \lambda \sqrt{p \log(p)} \right),
		\]
		which is $\oP(1)$ by our assumption.
		
		Finally, we prove $|\sqrt{n} \balpha_n^T \widetilde{\bTheta} \bDelta| = \oP(1)$. By the Cauchy-Schwartz inequality, $|\sqrt{n} \balpha_n^T \widetilde{\bTheta} \bDelta| \le \sqrt{n} \| \widetilde{\bTheta} \bDelta \|_2$, we only need that $\sqrt{n}\| \widetilde{\bTheta} \bDelta \|_2 = \oP(1)$. In equation (\ref{eq:taylor}), 
		\[
		\Delta_j  = \displaystyle \frac{1}{n} \sum_{i=1}^n  \left(  \ddot{\rho}(y_i, a_i^*) - \ddot{\rho}(y_i, \bx_i^T \widehat{\bxi})  \right) x_{ij} \bx_i^T (\bxi^0 - \widehat{\bxi}),
		\]
		where $a_i^*$ lies between $\bx_i^T\widehat{\bxi}$ and $\bx_i^T\bxi^0$, i.e. $ | a_i^* - \bx_i^T \widehat{\bxi} | \le | \bx_i^T ( \widehat{\bxi} - \bxi^0 ) | $. Then uniformly for all $j$,
		\[
		\begin{array}{rcl}
		| \Delta_j | & \le & \displaystyle \frac{1}{n} \sum_{i=1}^n  \vert \ddot{\rho}(y_i, a_i^*) - \ddot{\rho}(y_i, \bx_i^T \widehat{\bxi}) \vert \cdot |x_{ij}| \cdot  | \bx_i^T(\bxi^0 - \widehat{\bxi})| \\
		& \le & \displaystyle \frac{1}{n} \sum_{i=1}^n c_{Lip} | a_i^* - \bx_i^T \widehat{\bxi}   |  \cdot K \cdot | \bx_i^T(\bxi^0 - \widehat{\bxi})| \\
		& \le & \displaystyle c_{Lip}   K \cdot \frac{1}{n} \sum_{i=1}^n  | \bx_i^T(\bxi^0 - \widehat{\bxi})|^2 \\ 
		& = &  c_{Lip}   K \cdot  \OP(  s_0 \lambda^2) \\
		& = &  \OP(  s_0 \lambda^2),
		\end{array}
		\]
		where the last equality holds by Lemma 1. Since $\|  \bTheta_{\bxi^0}\| = \mathcal{O}(1)$ and $\| \widetilde{\bTheta} - \bTheta_{\bxi^0} \| = \oP(1)$, then $\| \widetilde{\bTheta} \| = \OP(1)$, and we  have
		\[
		\begin{array}{rcl}
		\sqrt{n} \| \widetilde{\bTheta} \bDelta \|_2 &  \le & \sqrt{n} \| \widetilde{\bTheta} \| \cdot \| \bDelta \|_{2} \\
		& \le & \sqrt{n} \OP(1) \cdot  \sqrt{p} \|\bDelta\|_{\infty} \\
		& \le & \OP( \sqrt{np}  s_0 \lambda^2 ).
		\end{array}
		\]
		By the assumption of $\sqrt{np}s_0\lambda^2 = o(1)$, $\sqrt{n}\|\widetilde{\bTheta}\bDelta\|_2 = \oP(1)$. 
		Applying Slutsky's Theorem and Lemma 3 gives the results.
	\end{proof}

	
	\bibliographystyle{apalike}
	
	\bibliography{thesis1-bib}

\begin{thebibliography}{}

\bibitem[Amos et~al., 2008]{amos2008genome}
Amos, C.~I., Wu, X., Broderick, P., Gorlov, I.~P., Gu, J., Eisen, T., Dong, Q.,
  Zhang, Q., Gu, X., Vijayakrishnan, J., et~al. (2008).
\newblock Genome-wide association scan of tag {SNP}s identifies a
  susceptibility locus for lung cancer at 15q25. 1.
\newblock {\em Nature Genetics}, 40(5):616--622.

\bibitem[B{\"u}hlmann and van~de Geer, 2011]{buhlmann2011statistics}
B{\"u}hlmann, P. and van~de Geer, S. (2011).
\newblock {\em Statistics for high-dimensional data: {Methods}, theory and
  applications}.
\newblock Heidelberg: Springer.

\bibitem[Cai et~al., 2019]{cai2019sparse}
Cai, T.~T., Zhang, A., and Zhou, Y. (2019).
\newblock Sparse group lasso: {Optimal} sample complexity, convergence rate,
  and statistical inference.
\newblock {\em arXiv preprint arXiv:1909.09851}.

\bibitem[Candes and Tao, 2007]{candes2007dantzig}
Candes, E. and Tao, T. (2007).
\newblock The {Dantzig} selector: {Statistical} estimation when $p$ is much
  larger than $n$.
\newblock {\em The Annals of Statistics}, 35(6):2313--2351.

\bibitem[Dezeure et~al., 2017]{dezeure2017high}
Dezeure, R., B{\"u}hlmann, P., and Zhang, C.-H. (2017).
\newblock High-dimensional simultaneous inference with the bootstrap.
\newblock {\em Test}, 26(4):685--719.

\bibitem[Doyle et~al., 2011]{doyle2011vitro}
Doyle, G.~A., Wang, M.-J., Chou, A.~D., Oleynick, J.~U., Arnold, S.~E., Buono,
  R.~J., Ferraro, T.~N., and Berrettini, W.~H. (2011).
\newblock \textit{In Vitro} and \textit{Ex Vivo} analysis of \textit{CHRNA3}
  and \textit{CHRNA5} haplotype expression.
\newblock {\em PloS One}, 6(8):e23373.

\bibitem[Eftekhari et~al., 2019]{eftekhari2019inference}
Eftekhari, H., Banerjee, M., and Ritov, Y. (2019).
\newblock Inference in general single-index models under high-dimensional
  symmetric designs.
\newblock {\em arXiv preprint arXiv:1909.03540}.

\bibitem[Evans and Relling, 2004]{evans2004moving}
Evans, W.~E. and Relling, M.~V. (2004).
\newblock Moving towards individualized medicine with pharmacogenomics.
\newblock {\em Nature}, 429(6990):464--468.

\bibitem[Fan and Li, 2001]{fan2001variable}
Fan, J. and Li, R. (2001).
\newblock Variable selection via nonconcave penalized likelihood and its oracle
  properties.
\newblock {\em Journal of the American Statistical Association},
  96(456):1348--1360.

\bibitem[Fan and Peng, 2004]{fan2004nonconcave}
Fan, J. and Peng, H. (2004).
\newblock Nonconcave penalized likelihood with a diverging number of
  parameters.
\newblock {\em The Annals of Statistics}, 32(3):928--961.

\bibitem[Friedman et~al., 2010]{friedman2010regularization}
Friedman, J., Hastie, T., and Tibshirani, R. (2010).
\newblock Regularization paths for generalized linear models via coordinate
  descent.
\newblock {\em Journal of Statistical Software}, 33(1):1--22.

\bibitem[Gabrielsen et~al., 2013]{gabrielsen2013association}
Gabrielsen, M.~E., Romundstad, P., Langhammer, A., Krokan, H.~E., and Skorpen,
  F. (2013).
\newblock Association between a 15q25 gene variant, nicotine-related habits,
  lung cancer and {COPD} among 56307 individuals from the {HUNT} study in
  {Norway}.
\newblock {\em European Journal of Human Genetics}, 21(11):1293--1299.

\bibitem[Guan and Stephens, 2011]{guan2011bayesian}
Guan, Y. and Stephens, M. (2011).
\newblock Bayesian variable selection regression for genome-wide association
  studies and other large-scale problems.
\newblock {\em The Annals of Applied Statistics}, 5(3):1780--1815.

\bibitem[Halld{\'e}n et~al., 2016]{hallden2016gene}
Halld{\'e}n, S., Sj{\"o}gren, M., Hedblad, B., Engstr{\"o}m, G., Hamrefors, V.,
  Manjer, J., and Melander, O. (2016).
\newblock Gene variance in the nicotinic receptor cluster
  {(CHRNA5-CHRNA3-CHRNB4)} predicts death from cardiopulmonary disease and
  cancer in smokers.
\newblock {\em Journal of Internal Medicine}, 279(4):388--398.

\bibitem[He and Lin, 2010]{he2010variable}
He, Q. and Lin, D.-Y. (2010).
\newblock A variable selection method for genome-wide association studies.
\newblock {\em Bioinformatics}, 27(1):1--8.

\bibitem[Hung et~al., 2008]{hung2008susceptibility}
Hung, R.~J., McKay, J.~D., Gaborieau, V., Boffetta, P., Hashibe, M., Zaridze,
  D., Mukeria, A., Szeszenia-Dabrowska, N., Lissowska, J., Rudnai, P., et~al.
  (2008).
\newblock A susceptibility locus for lung cancer maps to nicotinic
  acetylcholine receptor subunit genes on 15q25.
\newblock {\em Nature}, 452(7187):633--637.

\bibitem[Javanmard and Montanari, 2014]{javanmard2014confidence}
Javanmard, A. and Montanari, A. (2014).
\newblock Confidence intervals and hypothesis testing for high-dimensional
  regression.
\newblock {\em Journal of Machine Learning Research}, 15(1):2869--2909.

\bibitem[Lee et~al., 2016]{lee2016exact}
Lee, J.~D., Sun, D.~L., Sun, Y., and Taylor, J.~E. (2016).
\newblock Exact post-selection inference, with application to the lasso.
\newblock {\em The Annals of Statistics}, 44(3):907--927.

\bibitem[McKay et~al., 2017]{mckay2017large}
McKay, J.~D., Hung, R.~J., Han, Y., Zong, X., Carreras-Torres, R., Christiani,
  D.~C., Caporaso, N.~E., Johansson, M., Xiao, X., Li, Y., et~al. (2017).
\newblock Large-scale association analysis identifies new lung cancer
  susceptibility loci and heterogeneity in genetic susceptibility across
  histological subtypes.
\newblock {\em Nature Genetics}, 49(7):1126--1132.

\bibitem[Meinshausen and B{\"u}hlmann, 2006]{meinshausen2006high}
Meinshausen, N. and B{\"u}hlmann, P. (2006).
\newblock High-dimensional graphs and variable selection with the lasso.
\newblock {\em The Annals of Statistics}, 34(3):1436--1462.

\bibitem[Mitra and Zhang, 2016]{mitra2016benefit}
Mitra, R. and Zhang, C.-H. (2016).
\newblock The benefit of group sparsity in group inference with de-biased
  scaled group lasso.
\newblock {\em Electronic Journal of Statistics}, 10(2):1829--1873.

\bibitem[Ning and Liu, 2017]{ning2017general}
Ning, Y. and Liu, H. (2017).
\newblock A general theory of hypothesis tests and confidence regions for
  sparse high dimensional models.
\newblock {\em The Annals of Statistics}, 45(1):158--195.

\bibitem[Pintarelli et~al., 2017]{pintarelli2017pharmacogenetic}
Pintarelli, G., Galvan, A., Pozzi, P., Noci, S., Pasetti, G., Sala, F.,
  Pastorino, U., Boffi, R., and Colombo, F. (2017).
\newblock Pharmacogenetic study of seven polymorphisms in three nicotinic
  acetylcholine receptor subunits in smoking-cessation therapies.
\newblock {\em Scientific Reports}, 7(1):16730.

\bibitem[Repapi et~al., 2010]{repapi2010genome}
Repapi, E., Sayers, I., Wain, L.~V., Burton, P.~R., Johnson, T., Obeidat, M.,
  Zhao, J.~H., Ramasamy, A., Zhai, G., Vitart, V., et~al. (2010).
\newblock Genome-wide association study identifies five loci associated with
  lung function.
\newblock {\em Nature genetics}, 42(1):36.

\bibitem[Stevens et~al., 2008]{stevens2008nicotinic}
Stevens, V.~L., Bierut, L.~J., Talbot, J.~T., Wang, J.~C., Sun, J., Hinrichs,
  A.~L., Thun, M.~J., Goate, A., and Calle, E.~E. (2008).
\newblock Nicotinic receptor gene variants influence susceptibility to heavy
  smoking.
\newblock {\em Cancer Epidemiology and Prevention Biomarkers},
  17(12):3517--3525.

\bibitem[Sur and Cand{\`e}s, 2019]{sur2019modern}
Sur, P. and Cand{\`e}s, E.~J. (2019).
\newblock A modern maximum-likelihood theory for high-dimensional logistic
  regression.
\newblock {\em Proceedings of the National Academy of Sciences},
  116(29):14516--14525.

\bibitem[Taylor et~al., 2001]{taylor2001using}
Taylor, J.~G., Choi, E.-H., Foster, C.~B., and Chanock, S.~J. (2001).
\newblock Using genetic variation to study human disease.
\newblock {\em Trends in molecular medicine}, 7(11):507--512.

\bibitem[Thorgeirsson et~al., 2008]{thorgeirsson2008variant}
Thorgeirsson, T.~E., Geller, F., Sulem, P., Rafnar, T., Wiste, A., Magnusson,
  K.~P., Manolescu, A., Thorleifsson, G., Stefansson, H., Ingason, A., et~al.
  (2008).
\newblock A variant associated with nicotine dependence, lung cancer and
  peripheral arterial disease.
\newblock {\em Nature}, 452(7187):638--642.

\bibitem[Tibshirani, 1996]{tibshirani1996regression}
Tibshirani, R. (1996).
\newblock Regression shrinkage and selection via the lasso.
\newblock {\em Journal of the Royal Statistical Society. Series B
  (Methodological)}, 58(1):267--288.

\bibitem[van~de Geer et~al., 2014]{van2014asymptotically}
van~de Geer, S., B{\"u}hlmann, P., Ritov, Y., and Dezeure, R. (2014).
\newblock On asymptotically optimal confidence regions and tests for
  high-dimensional models.
\newblock {\em The Annals of Statistics}, 42(3):1166--1202.

\bibitem[van~de Geer, 2008]{van2008high}
van~de Geer, S.~A. (2008).
\newblock High-dimensional generalized linear models and the lasso.
\newblock {\em The Annals of Statistics}, 36(2):614--645.

\bibitem[Vershynin, 2010]{vershynin2010introduction}
Vershynin, R. (2010).
\newblock Introduction to the non-asymptotic analysis of random matrices.
\newblock {\em arXiv preprint arXiv:1011.3027}.

\bibitem[Vershynin, 2012]{vershynin2012close}
Vershynin, R. (2012).
\newblock How close is the sample covariance matrix to the actual covariance
  matrix?
\newblock {\em Journal of Theoretical Probability}, 25(3):655--686.

\bibitem[Wang, 2011]{wang2011gee}
Wang, L. (2011).
\newblock {GEE} analysis of clustered binary data with diverging number of
  covariates.
\newblock {\em The Annals of Statistics}, 39(1):389--417.

\bibitem[Zhang and Zhang, 2014]{zhang2014confidence}
Zhang, C.-H. and Zhang, S.~S. (2014).
\newblock Confidence intervals for low dimensional parameters in high
  dimensional linear models.
\newblock {\em Journal of the Royal Statistical Society: Series B (Statistical
  Methodology)}, 76(1):217--242.

\bibitem[Zhang and Cheng, 2017]{zhang2017simultaneous}
Zhang, X. and Cheng, G. (2017).
\newblock Simultaneous inference for high-dimensional linear models.
\newblock {\em Journal of the American Statistical Association},
  112(518):757--768.

\bibitem[Zou, 2006]{zou2006adaptive}
Zou, H. (2006).
\newblock The adaptive lasso and its oracle properties.
\newblock {\em Journal of the American Statistical Association},
  101(476):1418--1429.

\bibitem[Zou and Hastie, 2005]{zou2005regularization}
Zou, H. and Hastie, T. (2005).
\newblock Regularization and variable selection via the elastic net.
\newblock {\em Journal of the Royal Statistical Society: Series B (Statistical
  Methodology)}, 67(2):301--320.

\end{thebibliography}

	
	\clearpage
	
	\begin{landscape}

		\begin{table}[t]
			
			\caption{The association between SNPs and lung cancer risk in stratified analysis \\
				\textit{\scriptsize Coefficient estimates in logistic regression models are reported for demographic variables and 11 SNPs (a) among the smokers, and (b) among the non-smokers. The other SNPs are omitted from the table. ``Pos": physical location of a SNP on a chromosome (Assembly GRCh37/hg19); ``Est": estimated coefficient in the logistic regression models for the overall risk of lung cancer; ``SE": estimated standard error; ``CI": confidence interval.}}
			
			\renewcommand\arraystretch{0.6}
			\centering
			\subfloat[\footnotesize Smokers \label{tab:blcs_smo}]{
				\centering
				\scriptsize
				\begin{tabular}{ccccccccccccc}
					\hline
					& & & &   \multicolumn{3}{c}{\textit{REF-DS}} & \multicolumn{3}{c}{\textit{ORIG-DS}} & \multicolumn{3}{c}{\textit{MLE}}  \\
					\cmidrule(lr){5-7}  \cmidrule(lr){8-10}  \cmidrule(l){11-13} 
					\multicolumn{4}{r}{Demographic variable} & Est & SE & 95\% CI  & Est & SE & 95\% CI & Est & SE & 95\% CI   \\ \hline
					\multicolumn{4}{r}{Education: No high school} & 0.44 & 0.20 & (0.05, 0.83) & 0.48 & 0.19 & (0.11, 0.85) & 0.52 & 0.22 & (0.09, 0.94) \\ 
					\multicolumn{4}{r}{Education: High school graduate} & 0.00 & 0.15 & (-0.29, 0.28) & -0.02 & 0.14 & (-0.29, 0.25) & -0.01 & 0.16 & (-0.32, 0.30) \\ 
					\multicolumn{4}{r}{Gender: Male} & -0.14 & 0.13 & (-0.40, 0.12) & -0.16 & 0.13 & (-0.40, 0.09) & -0.15 & 0.14 & (-0.43, 0.13) \\ 
					\multicolumn{4}{r}{Age in years} & 0.04 & 0.07 & (-0.09, 0.17) & 0.05 & 0.06 & (-0.07, 0.17) & 0.05 & 0.07 & (-0.09, 0.19) \\ 
					\hline
					SNP & Pos & Allele & Gene & Est & SE & 95\% CI  & Est & SE & 95\% CI & Est & SE & 95\% CI   \\ 
					\hline
					AX-15319183 & 6:167352075 & C/G & \textit{RNASET2} & 0.01 & 0.19 & (-0.36, 0.39) & -0.03 & 0.18 & (-0.39, 0.33) & 0.02 & 0.20 & (-0.38, 0.42) \\ 
					AX-41911849 & 6:167360724 & A/G & \textit{RNASET2} & 0.43 & 0.22 & (0.00, 0.86) & 0.44 & 0.20 & (0.06, 0.83) & 0.49 & 0.24 & (0.03, 0.96) \\ 
					AX-42391645 & 8:27319769 & G/C & \textit{CHRNA2} & 0.01 & 0.16 & (-0.29, 0.32) & -0.01 & 0.14 & (-0.28, 0.26) & 0.01 & 0.16 & (-0.31, 0.34) \\ 
					\textbf{AX-38419741} & 8:27319847 & T/A & \textit{CHRNA2} & \textbf{0.11} & 0.35 & (-0.59, 0.80) & \textbf{-0.14} & 0.31 & (-0.75, 0.48) & 0.13 & 0.37 & (-0.60, 0.86) \\ 
					\textbf{AX-15934253} & 8:27334098 & T/C & \textit{CHRNA2} & \textbf{-0.15} & 0.44 & (-1.02, 0.71) & \textbf{0.06} & 0.39 & (-0.70, 0.82) & -0.19 & 0.47 & (-1.12, 0.74) \\ 
					AX-12672764 & 13:32927894 & T/C & \textit{BRCA2} & -0.07 & 0.19 & (-0.44, 0.31) & -0.10 & 0.16 & (-0.40, 0.21) & -0.07 & 0.20 & (-0.47, 0.33) \\ 
					AX-31620127 & 15:48016563 & C/T & \textit{SEMA6D} & 0.79 & 0.26 & (0.28, 1.31) & 0.79 & 0.25 & (0.30, 1.28) & 0.96 & 0.30 & (0.37, 1.55) \\ 
					\textbf{AX-88891100} & 15:78857896 & A/T & \textit{CHRNA5} & 0.87 & 0.36 & (0.17, 1.57) & 0.79 & 0.32 & (0.16, 1.41) & 0.98 & 0.39 & (0.22, 1.74) \\ 
					\textbf{AX-39952685} & 15:78867042 & G/C & \textit{CHRNA5} & 0.99 & 0.47 & (0.07, 1.91) & 0.82 & 0.38 & (0.09, 1.56) & 1.11 & 0.50 & (0.13, 2.08) \\ 
					AX-62479186 & 15:78878565 & T/C & \textit{CHRNA5} & 0.41 & 0.41 & (-0.40, 1.22) & 0.46 & 0.39 & (-0.31, 1.23) & 0.46 & 0.45 & (-0.41, 1.33) \\ 
					AX-88907114 & 19:41353727 & T/C & \textit{CYP2A6} & 0.52 & 0.34 & (-0.16, 1.19) & 0.49 & 0.33 & (-0.15, 1.13) & 0.58 & 0.38 & (-0.15, 1.32) \\ 
					$\vdots$ & & & & & & & & & & & &  \\
					\hline
				\end{tabular}
			}
			
			\subfloat[\footnotesize Non-smokers \label{tab:blcs_nonsmo}]{
				\centering
				\scriptsize
				\begin{tabular}{ccccccccccccc}
					\hline
					& & & &   \multicolumn{3}{c}{\textit{REF-DS}} & \multicolumn{3}{c}{\textit{ORIG-DS}} & \multicolumn{3}{c}{\textit{MLE}}  \\ 
					\cmidrule(lr){5-7}  \cmidrule(lr){8-10}  \cmidrule(l){11-13} 
					\multicolumn{4}{r}{Demographic variable} & Est & SE & 95\% CI  & Est & SE & 95\% CI & Est & SE & 95\% CI   \\ \hline
					\multicolumn{4}{r}{Education: No high school} & -0.84 & 0.93 & (-2.67, 0.99) & -0.58 & 0.78 & (-2.11, 0.95) & -10.53 & 3.82 & (-18.01, -3.04) \\ 
					\multicolumn{4}{r}{Education: High school graduate} & -1.68 & 0.52 & (-2.69, -0.66) & -1.56 & 0.43 & (-2.39, -0.72) & -11.23 & 3.75 & (-18.58, -3.88) \\ 
					\multicolumn{4}{r}{Gender: Male} & -0.30 & 0.41 & (-1.10, 0.51) & -0.16 & 0.32 & (-0.78, 0.46) & -1.94 & 1.03 & (-3.96, 0.09) \\ 
					\multicolumn{4}{r}{Age in years} & -0.52 & 0.20 & (-0.91, -0.13) & -0.56 & 0.16 & (-0.87, -0.26) & -2.59 & 0.98 & (-4.51, -0.67) \\ 
					\hline
					SNP & Pos & Allele & Gene & Est & SE & 95\% CI  & Est & SE & 95\% CI & Est & SE & 95\% CI   \\
					\hline
					AX-15319183 & 6:167352075 & C/G & \textit{RNASET2} & -0.71 & 0.55 & (-1.78, 0.36) & 0.01 & 0.40 & (-0.79, 0.80) & -4.32 & 1.84 & (-7.92, -0.71) \\ 
					AX-41911849 & 6:167360724 & A/G & \textit{RNASET2} & 0.69 & 0.65 & (-0.59, 1.97) & 0.37 & 0.47 & (-0.55, 1.29) & 4.46 & 2.00 & (0.54, 8.39) \\ 
					\textbf{AX-42391645} & 8:27319769 & G/C & \textit{CHRNA2} & \textbf{-0.11} & 0.49 & (-1.07, 0.85) & \textbf{0.18} & 0.30 & (-0.41, 0.78) & -1.90 & 2.00 & (-5.81, 2.02) \\ 
					AX-38419741 & 8:27319847 & T/A & \textit{CHRNA2} & 0.50 & 1.04 & (-1.54, 2.53) & 0.23 & 0.61 & (-0.97, 1.42) & 3.37 & 3.00 & (-2.51, 9.26) \\ 
					AX-15934253 & 8:27334098 & T/C & \textit{CHRNA2} & 0.11 & 1.40 & (-2.64, 2.86) & 0.38 & 0.82 & (-1.23, 1.98) & 5.37 & 4.21 & (-2.88, 13.62) \\ 
					AX-12672764 & 13:32927894 & T/C & \textit{BRCA2} & -0.83 & 0.62 & (-2.04, 0.37) & -0.57 & 0.38 & (-1.32, 0.18) & -8.25 & 2.64 & (-13.42, -3.08) \\ 
					\textbf{AX-31620127} & 15:48016563 & C/T & \textit{SEMA6D} & 1.77 & 0.75 & \textbf{(0.30, 3.24)} & 0.43 & 0.46 & \textbf{(-0.48, 1.34)} & 9.23 & 3.27 & (2.81, 15.64) \\ 
					AX-88891100 & 15:78857896 & A/T & \textit{CHRNA5} & 0.78 & 1.18 & (-1.54, 3.10) & 1.15 & 0.87 & (-0.56, 2.85) & 1.54 & 3.17 & (-4.68, 7.75) \\ 
					AX-39952685 & 15:78867042 & G/C & \textit{CHRNA5} & -0.54 & 1.30 & (-3.09, 2.01) & -0.99 & 0.73 & (-2.41, 0.44) & -2.85 & 3.98 & (-10.65, 4.96) \\ 
					\textbf{AX-62479186} & 15:78878565 & T/C & \textit{CHRNA5} & -1.28 & 1.34 & (-3.92, 1.35) & -1.33 & 1.10 & (-3.49, 0.82) & \textbf{-19.64} & \textbf{3410.98} & (-6705.04, 6665.75) \\ 
					\textbf{AX-88907114} & 19:41353727 & T/C & \textit{CYP2A6} & 0.86 & 0.88 & \textbf{(-0.86, 2.59)} & 1.40 & 0.68 & \textbf{(0.06, 2.74)} & 3.52 & 2.18 & (-0.75, 7.78) \\ 
					$\vdots$ & & & & & & & & & & & &  \\
					\hline
				\end{tabular}
			}
			
			\label{tab:blcs_strat1}
		\end{table}
		
	\end{landscape}

	\begin{figure}[h!] 
		\centering
		\includegraphics[width=0.8\textwidth]{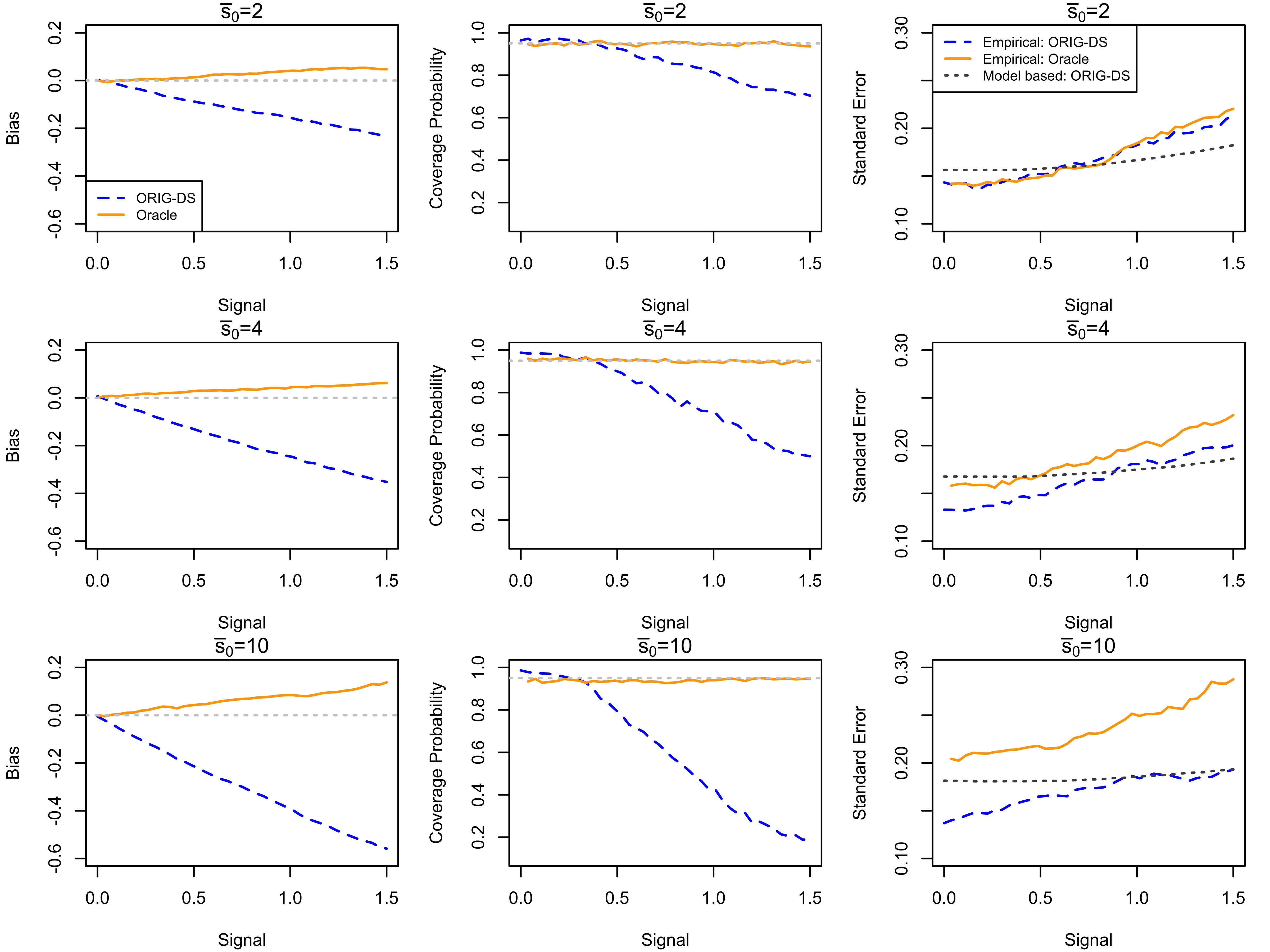}
		\caption{Simulation results of logistic regression with sample size $n=300$ and $p=500$ covariates. Covariates are first generated from multivariate Gaussian distribution with mean zero, AR(1) covariance structure and correlation 0.7, and truncated at $\pm 6$.  Each row presents estimation bias, empirical coverage probability and standard error (both model-based and empirical) of the estimated $\beta_1^0$, with 2, 4 and 10 additional signals fixed at 1 from the top to the bottom, respectively. ``ORIG-DS" and ``Oracle" stand for the original de-biased lasso estimator and the oracle estimator as if the true model were known, respectively.}
		\label{fig:sim_largep}
	\end{figure}

	\begin{landscape}
		\begin{figure}
			\centering
			\includegraphics[height=0.8\textheight]{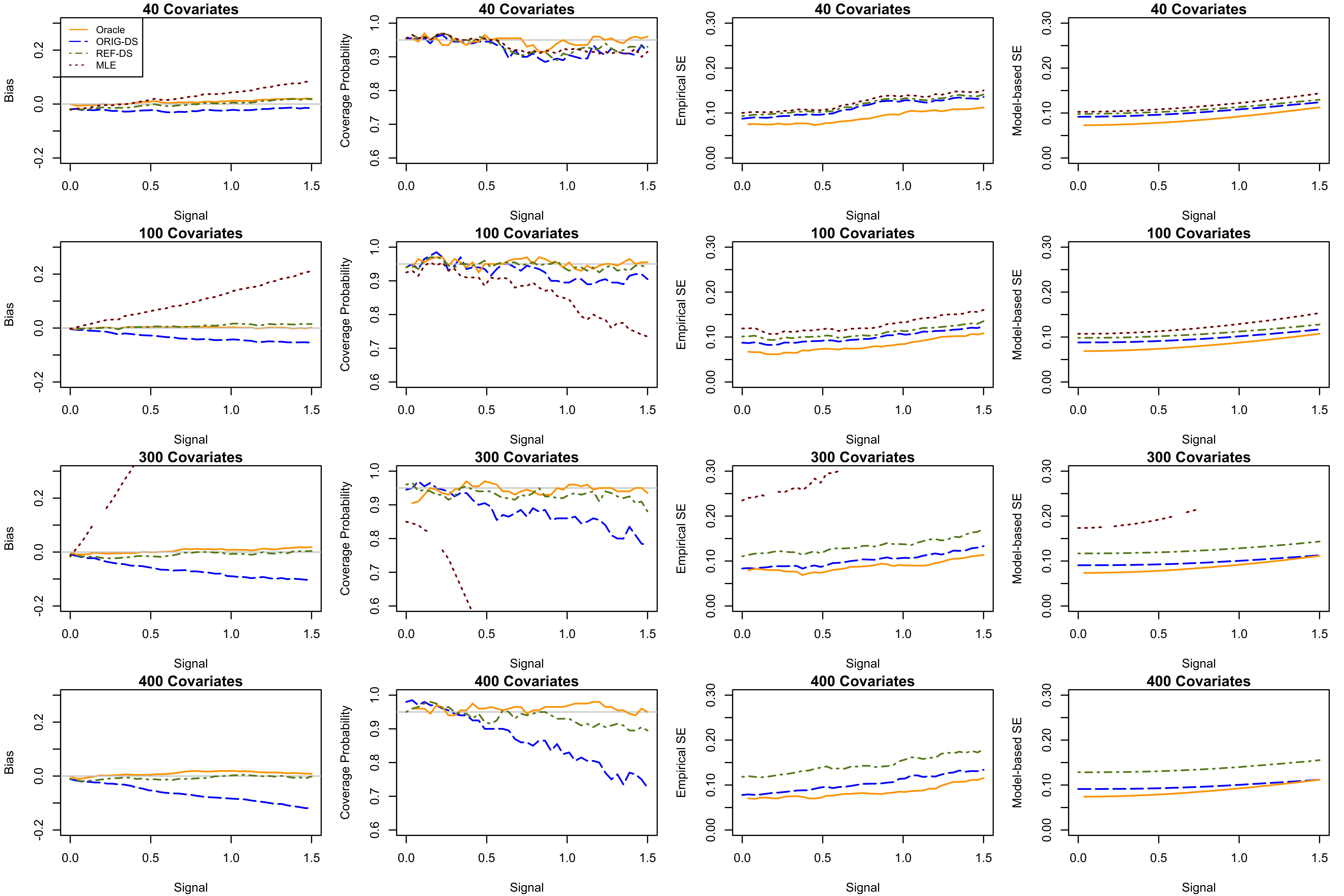}
			\caption{Simulation results: Bias, coverage probability, empirical standard error, and model-based standard error for $\beta_1^0$ in a logistic regression. Covariates are simulated with $\Sigma_x$ being AR(1) with $\rho=0.7$. The sample size is $n=1,000$ and the number of covariates $p = 40, 100, 300, 400$. MLE under the true model, denoted as ``\textit{Oracle}", is plotted as a reference.}
			\label{fig:logit_n1k_ar1}
		\end{figure}
	\end{landscape}

\end{document}